\documentclass[journal]{IEEEtran}
\usepackage[numbers,sort&compress]{natbib}
\usepackage{amsmath,amssymb,amsfonts,amsthm}
\usepackage{algpseudocode}
\usepackage{algorithmicx,algorithm}
\usepackage{graphicx}
\usepackage[caption=false,font=normalsize,labelfont=sf,textfont=sf]{subfig}
\usepackage{xcolor}
\usepackage{bm}
\usepackage{stfloats}
\usepackage{cleveref}
\usepackage{empheq}
\usepackage[theorems,skins]{tcolorbox}
\usepackage{ragged2e}
\newcounter{TempEqCnt}
\begin{document}
\newtheorem{proposition}{Proposition}
\newtheorem{lemma}{Lemma}

\title{Beam Pattern Modulation Embedded Hybrid Transceiver Optimization for Integrated Sensing and Communication}

\author{Boxun Liu,~\IEEEmembership{Graduate Student Member,~IEEE,} Shijian Gao,~\IEEEmembership{Member, IEEE,} Zonghui Yang,~\IEEEmembership{Graduate Student Member,~IEEE,} Xiang Cheng,~\IEEEmembership{Fellow,~IEEE,} Liuqing Yang,~\IEEEmembership{Fellow,~IEEE}

\thanks{{Part of this paper has been presented at the 2024-Spring IEEE Vehicular Technology Conference (VTC2024-spring, Singapore) \cite{liu2024beam}}.

Boxun Liu, Zonghui Yang, and Xiang Cheng are with the State Key Laboratory of Advanced Optical Communication Systems and Networks, School of Electronics, Peking University, Beijing 100871, China (e-mail:boxunliu@stu.pku.edu.cn; yzh22@stu.pku.edu.cn;xiangcheng@pku.edu.cn).

Shijian Gao is with the Internet of Things Thrust, The Hong Kong University of Science and Technology (Guangzhou), Guangzhou 511400, China (e-mail: shijiangao@hkust-gz.edu.cn).

Liuqing Yang is with the Internet of Things Thrust and Intelligent Transportation Thrust, The Hong Kong University of Science and Technology (Guangzhou), Guangzhou 511400, China, and also with the Department
of Electronic and Computer Engineering and the Department of Civil and Environmental Engineering, The Hong Kong University of Science and Technology, Hong Kong, SAR, China (e-mail: lqyang@ust.hk).
}}

\markboth{}%
{Shell \MakeLowercase{\textit{et al.}}: A Sample Article Using IEEEtran.cls for IEEE Journals}
\maketitle
\begin{abstract}
Integrated sensing and communication (ISAC) emerges as a promising technology for 6G, particularly in the millimeter-wave (mmWave) band. 
However, the widely utilized hybrid architecture in mmWave systems compromises multiplexing gain due to the constraints of limited radio-frequency (RF) chains.
Moreover, additional sensing functionalities exacerbate the impairment of spectrum efficiency (SE). 
In this paper, we present an optimized beam pattern modulation-embedded ISAC (BPM-ISAC) transceiver design, which spares one RF chain for sensing and {uses the remaining ones for communication}.
To compensate for the reduced SE, index modulation across communication beams is applied.
We formulate an optimization problem aimed at minimizing the mean squared error (MSE) of the sensing beampattern, subject to a symbol MSE constraint. 
This problem is then solved by sequentially optimizing the analog and digital parts.
Both the multi-aperture structure (MAS) and the multi-beam structure (MBS) are considered in the analog part.
We conduct theoretical analysis on the asymptotic pairwise error probability (APEP) and the Cramér-Rao bound (CRB) of direction of arrival (DoA) estimation.
Numerical simulations validate the overall enhanced ISAC performance over existing alternatives.
\end{abstract}


\section{Introduction}
\IEEEPARstart{I}{ntegrated} sensing and communications (ISAC) \cite{liu2023seventy,cheng2022integrated,cheng2023intelligent,fan2022radar} is a pivotal technology for B5G/6G, striving for symbiosis and mutual enhancement of communication and sensing with sharing resources such as spectrum, hardware, and energy. 
Recently, millimeter-wave (mmWave) ISAC has gained substantial attention due to its broader bandwidth, facilitating higher data rates and improved detection accuracy for both communication and sensing.
Moreover, sensing and communication share similar channel characteristics and signal processing techniques in the mmWave frequency band \cite{liu2022integrated}, further enabling their seamless integration.

Transceiver design is vital for mmWave ISAC system, aiming to realize better performance trade-offs between communication and sensing.
{A large proportion of ISAC transceivers \cite{huang2020majorcom,ma2021spatial,xu2022hybrid} rely on fully digital (FD) architectures, making them impractical to deploy in mmWave ISAC massive MIMO systems due to high hardware costs and power consumption. 
To address this issue, some studies have explored low-cost hybrid architectures \cite{gao2018low} for mmWave ISAC transceiver design \cite{liu2019hybrid,wang2022partially,zhang2018multibeam,zhuo2022multibeam,gao2022integrated}, where the number of RF chains is fewer than the number of antennas.
In \cite{liu2019hybrid}, a fully-connected hybrid transceiver design was proposed for the single-user mmWave ISAC scenario by approximating the optimal communication and radar precoder. 
To further lower the hardware cost, the partially-connected hybrid transceiver architecture \cite{wang2022partially} has been adopted for enabling multi-user ISAC, which minimizes the Cramér-Rao bound (CRB) for direction of arrival (DoA) estimation under communication constraints. 
However, the spectral efficiency (SE) of hybrid ISAC systems is impaired due to two factors. 
On the one hand, the restricted number of RF chains damages the potential multiplexing gain (MG).
On the other hand, additional sensing functions will consume system resources, inevitably causing a further decrease in SE.}

{To achieve higher SE, index modulation (IM) has emerged as a promising technology for delivering additional information by selectively activating the state of certain resource domains \cite{younis2010generalised,ding2018beam,gao2019spatial}, such as antennas and subcarriers \cite{yang2024superposed}.
Recently, some sensing-centric ISAC transceiver designs have been proposed in conjunction with IM to improve SE \cite{ma2023index,huang2020majorcom,ma2021spatial,xu2022hybrid}. 
In \cite{huang2020majorcom}, a multi-carrier agile joint radar communication (MAJoRCom) system was proposed based on carrier agile phased array radar (CAESAR), where communication bits are transmitted through selective activation of radar waveforms on subcarriers and antennas. 
Furthermore, a hybrid index modulation (HIM) scheme was proposed \cite{xu2022hybrid} for frequency hopping MIMO radar communications system, where communication bits are transmitted through index modulation on entwined frequency, phase, and antenna tuples.
While radar functionality remains unaffected in \cite{huang2020majorcom,xu2022hybrid}, it results in a significantly low communication rate.
In \cite{ma2021spatial}, a spatial modulation-based communication-radar (SpaCoR) system was proposed, where individual sensing and communication waveforms are transmitted on different antennas, and generalized spatial modulation (GSM) is adopted to embed additional data bits through antenna selection.
Nevertheless, the data rates are limited by the radar pulse period.
In addition, these designs rely on antenna activation-based index modulation and are exclusively designed for FD architecture, limiting direct application to hybrid systems.}

To better cope with the hybrid structures, generalized beamspace modulation (GBM) was proposed in \cite{gao2019spatial}, which utilizes the unique sparsity of mmWave beamspace channel to elevate SE by implementing IM over beamspace.
However, it is designed for mmWave communication-only systems without sensing capabilities.
Built upon GBM, recent works \cite{guo2022non,elbir2023millimeter,elbir2023spatial} introduce similar IM schemes into mmWave ISAC systems to attain higher SE. 
In \cite{guo2022non}, the dual-functional beam pattern is selectively activated in a non-uniform manner for a higher SE. 
Nonetheless, the additional sensing capability of beam patterns inevitably compromises the communication performance.
In \cite{elbir2023millimeter}, the ISAC transmitter selectively activates partial spatial paths for communication and employs a single fixed beam for sensing, termed SPIM-ISAC.
The design of separate communication and sensing beams enhances communication SE while ensuring sensing performance.
The subsequent work \cite{elbir2023spatial} delves further into the consideration of the beam squint effect in the terahertz frequency band.
{In fact, SPIM is a special case of GBM without beam optimization, which constructs beamspace through fixed strongest spatial paths.
However, as it extends to non-line-of-sight (NLoS) scenarios, multi-angle sensing beams introduce potential disturbance to communication receivers due to the randomness of targets' angles, thereby deteriorating the error performance.
Additionally, SPIM-ISAC achieves a performance trade-off solely through power allocation between optimal communications-only and sensing-only beamformers, lacking a comprehensive consideration of overall performance.}

Considering the limitations highlighted in the previously mentioned works, we have developed a communication-centric mmWave ISAC transceiver design, where one dedicated RF chain is reserved for sensing. 
To address the decrease in SE resulting from the reduced number of RF chains allocated for communication, we have extended GBM to beam pattern modulation (BPM) for communication beams.
Compared to GBM based on the ideal beamspace domain and SPIM based on the channel path domain, BPM considers a more generalized beam pattern concept, where each beam pattern is formed by the corresponding column of analog precoders.
{Nevertheless, more flexible beam pattern and additional sensing requirements increase the complexity of the transceiver optimization.}
In light of the sensing interference on the communication receiver, we formulate a joint optimization problem to minimize the sensing beampattern mean squared error (MSE) under the symbol MSE constraint.
We solve it by optimizing analog and digital parts sequentially, where both the multi-aperture structure (MAS) and the multi-beam structure (MBS) are considered for analog part optimization.
For MBS, a low-complexity 2-step beam selection algorithm based on the min-MSE criterion is proposed.
For MAS, we adopt the branch and bound algorithm for analog sensing precoder design and the entry-wise iteration algorithm for analog communication parts design.
With the fixed analog part, the digital part is optimized using the proposed alternating optimization algorithm for improved power allocation.
Moreover, the communication and sensing performance are theoretically analyzed. \footnote{Simulation codes are provided to reproduce the results presented in this paper: https://github.com/liuboxun/BPM-ISAC}
The contributions of our work can be summarized as follows.

\begin{itemize}
\item We propose a beam pattern modulation-embedded hybrid transceiver design for mmWave ISAC systems (BPM-ISAC), where the ISAC transmitter provides multi beams for single-user communication and scanning beams for sensing, respectively.
The communication beams are selectively activated to enhance SE.
\item {We formulate an optimization problem to minimize sensing beampattern MSE with symbol MSE constraint and solve it via optimizing analog and digital parts sequentially.}
Two typical hybrid structures, namely MBS and MAS, are considered for the analog part design.
\item We conduct a theoretical analysis of the complexity and convergence of the proposed algorithm. 
The asymptotic pairwise error probability (APEP) and the CRB are derived to illustrate the bit error and DoA estimation performance. 
Additionally, simulation results validate the proposed method's advantages in ISAC.
\end{itemize}
\begin{figure*}[ht]
\center{\includegraphics[width=16cm]  {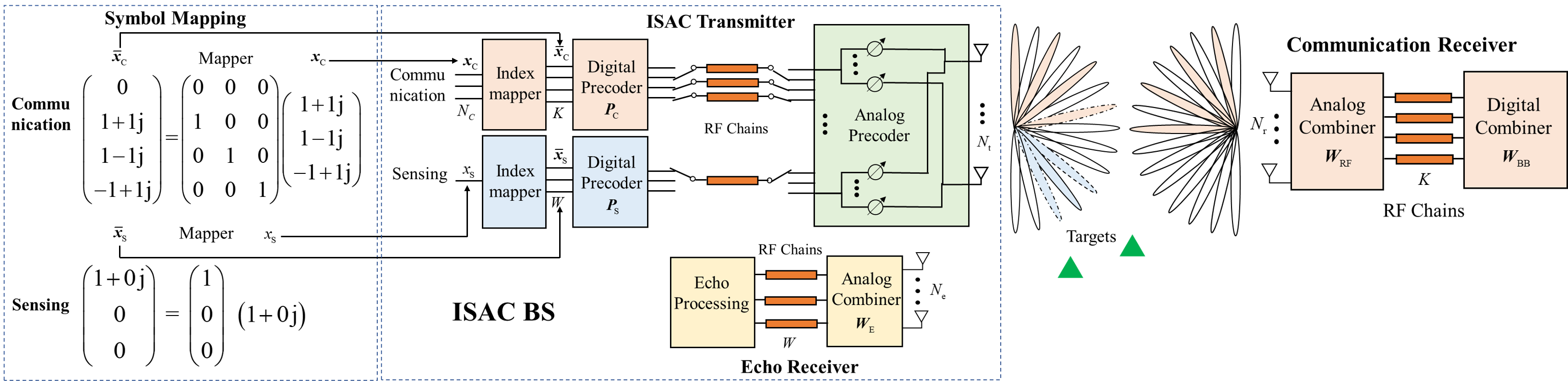}} 
\vspace{-2mm}
\caption{{An illustration of the symbol mapping scheme and the transceiver architecture for the proposed BPM-ISAC mmWave system. ($N_{\rm C}=3$,$K=4$,$W=3$)}}
\vspace{-4mm}
 \label{System}
\end{figure*}

The rest of this paper is organized as follows: Section \uppercase\expandafter{\romannumeral2}
 introduces the system and signal model of the proposed BPM-ISAC system. 
 Section \uppercase\expandafter{\romannumeral3} formulates an optimization problem and Section \uppercase\expandafter{\romannumeral4} proposes a joint hybrid transceiver design to solve it.
 Then Section  \uppercase\expandafter{\romannumeral5} gives the performance analysis and Section  \uppercase\expandafter{\romannumeral6} provides numerical simulations.
 Finally, Section  \uppercase\expandafter{\romannumeral7} concludes this paper.

\textit{Notation}:  $(\cdot)^{\rm T}$, $(\cdot)^{\rm H}$, $(\cdot)^\dagger$,  $\Vert\cdot\Vert_2$, and $\Vert\cdot\Vert_F$ denote the transpose, the conjugate transpose, pseudo-inverse, 2 norm and  Frobenius norm, respectively. 
{$\bm{a}[i]$ is the $i\mbox{-}$th element of a vector $\bm{a}$ and $\bm{A}[i,j]$ denotes the element of matrix $\bm{A}$ at the $i\mbox{-}$th row and the $j\mbox{-}$th column.} 
$\dot{\bm{a}}(\psi)= \frac{ \partial \bm{a}(\psi) }{\partial \psi}$ means the derivative of vector $\bm{a}$ over $\psi$.
$\mathcal{CN}(m,\sigma^2)$ represents the complex Gaussian distribution whose mean is $m$ and covariance is $\sigma^2$. 
{$\bm{F}_N$ denotes the $N\mbox{-}$dimensional discrete Fourier transform (DFT) matrix.}
$\bm{I}_K$ is the $K\times K$ identity matrix and $\bm{1}_K$ denotes the $K\times 1$ all-one column vector. ${\rm{diag}}(\bm{a})$ means diagonal matrix formed from vector $\bm{a}$ and $\mathbb{E}(\cdot)$ means expectation operation. $\mathbb{R}$ and $\mathbb{C}$  denote the set of real numbers and complex numbers, respectively.

\vspace{-2mm}
\section{System model}

As shown in Fig. \ref{System}, we consider a mmWave ISAC system for point-to-point communication and multi-target detection, which consists of an ISAC base station (BS) and a communication receiver. 
The ISAC BS comprises an ISAC transmitter and an echo receiver for simultaneous communication and monostatic sensing. 
In our modeling, the communication receiver and sensing targets are assumed to be spatially distinct.
Both the ISAC transmitter, the echo receiver, and the communication receiver adopt fully-connected hybrid architecture, equipped with $N_{\rm t}$, $N_{\rm e}$, and $N_{\rm r}$ half-wavelength spaced uniform linear antenna array, respectively.  
\vspace{-2mm}
\subsection{{Beam Pattern Modulation for ISAC}}
At the ISAC transmitter, $K$ communication beams and $W$ sensing scanning beams are generated through corresponding digital and analog precoders, i.e.,
\begin{align}
    \bm{s}=\bm{F}_{\rm C}\bm{P}_{\rm C}\bar{\bm{x}}_{\rm C}+\bm{F}_{\rm S}\bm{P}_{\rm S}\bar{\bm{x}}_{\rm S},
\end{align}
where $\bm{\bar{x}}_{\rm C} \in \mathbb{C}^{K \times 1}$ and $\bm{\bar{x}}_{\rm S}\in\mathbb{C}^{W \times 1}$ are the {mapped} communication and sensing symbols, respectively. 
$\bm{F}_{\rm C}\in\mathbb{C}^{N_{\rm t}\times K}$  and $\bm{F}_{\rm S}\in\mathbb{C}^{N_{\rm t}\times W}$ are analog precoders for communication and sensing, respectively.
{
$\bm{P}_{\rm C}={\rm{diag}}\left(\bm{p}\right)\in\mathbb{R}^{K \times K}$ and $\bm{P}_{\rm S}={\rm{diag}}\left(\bm{b}\right)\in\mathbb{R}^{W \times W}$ are the corresponding digital precoders, where each element of $\bm{p}$ and $\bm{b}$ represents the power allocation on the associated beam.} 
\subsubsection{{Communication}}
For communication, IM is implemented on the beam pattern domain. 
Specifically, in each symbol period, only $N_{\rm C}$ out of $K$ communication beams are activated.
To realize the selective activation, $N_{\rm C}$-dimensional non-zero data stream $\bm{x}_{\rm C} \in \mathbb{C}^{N_{\rm C}\times 1}$ is mapped to $K$-dimensional zero-containing $\bm{\bar{x}}_{\rm C}$ with totally ${\rm C}_K^{N_{\rm C}}$ possible index patterns.
Similar to \cite{gao2019spatial}, $2^{\lfloor {\rm log_{2}}{{\rm C}_K^{N_{\rm C}}} \rfloor}$ of these patterns are utilized to transmit additional  $\lfloor {\rm log_{2}}{{\rm C}_K^{N_{\rm C}}} \rfloor$ index bits.
Suppose M-ary phase shift keying/quadrature amplitude modulation (PSK/QAM) is adopted for communication, and the achievable SE is given by
\begin{align}\label{SE}
    \eta=N_{\rm C}\log_2{M}+\lfloor{\rm log_{2}}{{\rm C}_K^{N_{\rm C}}} \rfloor {\rm bps/Hz}.
\end{align}
\subsubsection{{Sensing}}
For sensing, {to save RF chains}, each of the $W$ beams is sequentially activated to scan $W$ directions of interest. 
Therefore, sensing signal $x_{\rm S}$ is mapped to a $W$-dimensional one-hot vector $\bm{\bar{x}}_{\rm S}\in\mathbb{C}^{W \times 1}$ before transmission.
For more flexible sensing, the case of non-equal probability scanning is considered.
Denote the activation probability matrix as $\bm{D}={\rm{diag}}([d_1,...,d_W])$, where $d_i$ represents the predefined activation probability of the $i\mbox{-}$th sensing beam and satisfies $\sum_{i=1}^Wd_i=1$.

For the hardware implementation, as shown in Fig. \ref{System}, only $N_{\rm C}$ and one RF chains are deployed at the ISAC transmitter for communication and sensing, respectively. 
{The switching network is controlled to adjust the RF chain connection for non-zeros symbols transmission.} 
In addition, we assume $K$ and $W$ RF chains are employed for the communication receiver and echo receiver, respectively.

\subsection{{I-O Relationship of Communication}}
The classical Saleh-Valenzuela mmWave channel \cite{saleh1987statistical} with $P$ dominant paths is adopted as
\begin{align}\label{channel} 
	\bm{H}=\sqrt{\frac{N_{\rm t}N_{\rm r}}{P}}\sum_{i = 1}^{P}\alpha_i\bm{a}_{N_{\rm r}}\left(\theta_i\right)\bm{a}_{N_{\rm t}}^{\rm H}\left(\phi_i\right), 
\end{align}
where $\alpha_i$ is the complex gain of the $i\mbox{-}$th path.  $\bm{a}_{N_{\rm r}}\left(\theta_i\right)$ and $\bm{a}_{N_{\rm t}}\left(\phi_i\right)$ are the  channel steering vectors of the $i\mbox{-}$th path, where {$\bm{a}_N(\theta)[i]=\frac{1}{\sqrt{N_{\rm t}}}e^{-j\pi(i-1) \sin(\theta)}$.} 
We assume that $\bm{H}$ is available at the transmitter.  
{In the time division duplex (TDD) system, this can be achieved via advanced channel estimation schemes \cite{gao2020estimating} in the uplink, while in the frequency division duplex (FDD), this can resort to downlink estimation accompanied by dedicated feedback strategies \cite{guo2022overview}.}

The received signal processed by analog combiner $\bm{W}_{\rm RF}\in\mathbb{C}^{N_{\rm r}\times K}$ becomes
\begin{align}
    \bm{y}_C =&\bm{W}^{\rm H}_{\rm RF}\bm{H}\bm{F}_{\rm C}\bm{P}_{\rm C}\bm{\bar{x}}_{\rm C}+\bm{W}^{\rm H}_{\rm RF}\bm{H}\bm{F}_{\rm S}\bm{P}_{\rm S}\bm{\bar{x}}_{\rm S}+\bm{\xi}_{\rm C}\nonumber\\
=&\bm{H}_{\rm C}\bm{P}_{\rm C}\bm{\bar{x}}_{\rm C}+\bm{H}_{\rm S}\bm{P}_{\rm S}\bm{\bar{x}}_{\rm S}+\bm{\xi}_{\rm C}, 
\end{align}
where $\bm{H}_{\rm C}$ and $\bm{H}_{\rm S}$ is the equivalent digital channel (EDC) for communication and sensing, and $\bm{\xi}_{\rm C}\sim\mathcal{CN}(0,\sigma^2 \bm{I}_K)$ is additive white Gaussian noise (AWGN). 
{It is noteworthy that the communication received signal is disturbed by sensing signal and noise.
There exists a trade-off between sensing and communication, wherein higher sensing power will increase the symbol error rate.
}

To eliminate the sensing interference, one possible method is to estimate the sensing signal and then subtract it at the communication receiver \cite{zhuo2022multibeam}. 
However, this scheme assumes that
{the instantaneous sensing signal} is known to the user, limiting the freedom degree of sensing waveform and increasing the operational complexity. 
In this paper, {we suppose that only second-order statistics of the sensing signal are known at the communication receiver.}
Therefore, the well-known LMMSE equalizer is adopted as
\begin{align} \label{m}
\bm{W}_{\rm BB}=&\bm{R}_{\bm{\bar{x}}_{\rm C}}\bm{P}_{\rm C}\bm{H}_{\rm C}^{\rm H}\left(\bm{H}_{\rm C}\bm{P}_{\rm C}\bm{R}_{\bm{\bar{x}}_{\rm C}}\bm{P}_{\rm C}\bm{H}_{\rm C}^{\rm H}+\right.\nonumber\\
&\left.\bm{H}_{\rm S}\bm{P}_{\rm S}\bm{R}_{\bm{\bar{x}}_{\rm S}}\bm{P}_{\rm S}\bm{H}_{\rm S}^{\rm H}+\sigma^2\bm{I}_K\right)^{-1},
\end{align}
where 
$\bm{R}_{\bm{\bar{x}}_{\rm C}}=\mathbb{E}\left[\bm{\bar{x}}_{\rm C}\bm{\bar{x}}_{\rm C}^{\rm H}\right]=\frac{N_{\rm C}}{K}\bm{I}_K$ and $R_{\bm{\bar{x}}_{\rm S}}=\mathbb{E}\left[\bm{\bar{x}}_{\rm S}\bm{\bar{x}}_{\rm S}^{\rm H}\right]=\bm{D}$.
Then the symbol $\bm{\bar{x}}_{\rm C}$ is estimated as   
 \begin{align}\label{esti sym} \tilde{\bm{x}}_C=\bm{W}_{\rm BB}\left(\bm{H}_{\rm C}\bm{P}_{\rm C}\bm{\bar{x}}_{\rm C}+\bm{H}_{\rm S}\bm{P}_{\rm S}\bm{\bar{x}}_{\rm S}+\bm{\xi}_{\rm C}\right).
 \end{align}
The information bits contained in $\bm{x}_{\rm C}$ and index bits can be estimated by the maximum likelihood (ML) detector or 2-step quantization detector \cite{gao2019spatial}, so the details are omitted here.

\subsection{{I-O Relationship of Sensing}}
Suppose there are $N$ pointed targets with $i\mbox{-}$th one locating at angles $\psi_i$. 
The ISAC transmitter sequentially transmits $W$ sensing beams to cover the range of interest.
Then the echoes are received to estimate the target parameters. 
Considering quasi-static sensing processes, it is equivalent to simultaneously scanning and receiving echoes from various directions. 
We assume that the self-interference on the echo receiver from the transmitter can be effectively mitigated \cite{liu2023joint}.
{Besides, the communication reflected signals at the target are ignored because they are relatively weak compared to the sensing signal \cite{zhuo2022multibeam}.}
{During per scanning}, the received echo signal is approximated as
\noindent
\begin{align}
\bm{y}_{\rm R}&=\sum_{i=1}^{N}\beta_i \bm{a}_{N_{\rm e}}\left(\psi_i\right)\bm{a}_{N_{\rm t}}^{\rm H}\left(\psi_i\right)\left(\bm{F}_{\rm S}\bm{P}_{\rm S}\bar{\bm{x}}_{\rm S}+\bm{F}_{\rm C}\bm{P}_{\rm C}\bar{\bm{x}}_{\rm C}\right)+\bm{\xi}_{\rm R}\nonumber\\
&\overset{(a)}{\simeq}\sum_{i=1}^{N}\beta_i\bm{a}_{N_{\rm e}}\left(\psi_i\right)\bm{a}_{N_{\rm t}}^{\rm H}(\psi_i)\bm{F}_{\rm S}\bm{P}_{\rm S}\bar{\bm{x}}_{\rm S}+\bm{\xi}_{\rm R},
\end{align}
where $\beta_i$ denotes the $i\mbox{-}$th reflection coefficient of the target and $\bm{\xi}_{\rm R}\in \mathbb{C}^{N_{\rm e}\times 1}$ is the additive white Gaussian noise.
(a) is for that the communication beam will be selected to stay away from the target direction to avoid sensing interference, and $\bm{a}_{N_{\rm t}}^{\rm H}(\psi_i)\bm{F}_{\rm C}\bm{P}_{\rm C}\bar{\bm{x}}_{\rm C}\simeq 0$ is satisfied. 
{For cases where the sensing targets are close to the communication receiver, sensing can be achieved simply using the communication beam without additional optimization, which is not within the scope of our study.}
Denote $\bm{\Xi}={\rm{diag}}\left(\beta_1, ..., \beta_N\right)$ and $\bm{A}_M=[\bm{a}_M\left(\psi_1\right),...,\bm{a}_M\left(\psi_N\right)]$, and {we can obtain the compact form as}
\begin{align}
\bm{y}_{\rm R}\simeq\bm{A}_{N_{\rm e}}\Xi\bm{A}_{N_{\rm t}}^{\rm H}\bm{F}_{\rm S}\bm{P}_{\rm S}\bar{\bm{x}}_{\rm S}+\bm{\xi}_{\rm R}.
\end{align}
For simplicity, we assume that $N_{\rm t}=N_{\rm r}$ and denote $\bm{A}=\bm{A}_{N_{\rm t}}=\bm{A}_{N_{\rm e}}$.
Since the direction of departure (DoD) and DoA of the target are the same, the sensing analog combiner $\bm{W}_{\rm E}$ can be implemented as $\bm{F}_{\rm S}$. Then the sensing baseband received signal is derived as 
 \begin{align}\label{sensing baseband}
     \bm{y}_{\rm B}&=\bm{F}_{\rm S}^{\rm H}\bm{A}\Xi\bm{A}^{\rm H}\bm{F}_{\rm S}\bm{P}_{\rm S}\bar{\bm{x}}_{\rm S}+\bm{F}_{\rm S}^{\rm H} \bm{\xi}_{\rm R}\nonumber\\
     &=\bm{T}_{\rm B}^{\rm H}\Xi\bm{T}_{\rm B}\bm{P}_{\rm S}\bar{\bm{x}}_{\rm S}+\bm{\xi}_{\rm B},
 \end{align}
where $\bm{T}_{\rm B}=\bm{A}^{\rm H}\bm{F}_{\rm S}$ and $\bm{\xi}_{\rm B}\sim\mathcal{CN}(0,\bm{R}_{\rm B})$. The target parameters can be estimated with existing algorithms. 
{For example, DoA estimation can be performed using the beamspace MUSIC algorithm \cite{lee1990resolution}.}
\vspace{-2mm}
\section{{Problem Formulation for BPM-ISAC}}
In this section, we establish the performance criterion of sensing and communication, and formulate a joint optimization problem to achieve the desired sensing beampattern with reliable communication via optimizing the hybrid transceivers.
\subsection{{Sensing Performance Criterion}}
For accurate parameter estimation, radiating sufficient energy in the directions of interest is crucial. 
Hence, the beampattern is adopted as the sensing performance metric, measuring the sensing power distribution in different directions.
Denote $\theta_t$ as the $t\mbox{-}$th direction of interest, then the actual sensing beampattern is defined as 
\begin{align}
	\bm{v}=[\lvert b_1 \bm{a}_{N_{\rm t}}^{\rm H}(\theta_1)\bm{F}_{\rm S}[:,1]\rvert,..,\lvert b_W \bm{a}_{N_{\rm t}}^{\rm H}(\theta_W)\bm{F}_{\rm S}[:,W]\rvert]^T,
\end{align}
where $b_i$ is the $i\mbox{-}$th element of $\bm{b}$ and represents the allocated power on $i\mbox{-}$th sensing beam.

To augment sensing performance, the actual sensing beampattern needs to maximally match the ideal beampattern. 
We predefine the ideal beampattern as $\bm{t}\in \mathbb{R}^{W\times1}$, which satisfies the power constraint $\Vert \bm{D}^\frac{1}{2}\bm{t}\Vert_2^2=\sum_{i=1}^{W}d_it_i^2=T_{\rm R}$, where $t_i$ is the $i\mbox{-}$th element of $\bm{t}$ and $T_{\rm R}$ is the average sensing power. 
The sensing beampattern MSE \cite{cheng2022qos} is adopted as the criterion of sensing performance, which measures the discrepancy between sensing beampattern $\bm{v}$ and ideal beampattern $\bm{t}$.
Considering the activation probability of each beam, it is derived as
\begin{align}
	{\rm MSE_S}(\bm{P}_{\rm S},\bm{F}_{\rm S})=\Vert \bm{D}^\frac{1}{2}(\bm{v}-\bm{t})\Vert_2^2=\sum_{i=1}^{W}d_i(v_i-t_i)^2.
\end{align}
\subsection{{Communication Performance Criterion}}
{The communication performance includes effectiveness and reliability, which are characterized by SE and symbol error, respectively.}
Since the SE of the proposed method is determined as Eq. (\ref{SE}), we resort to the symbol MSE under fixed SE as the performance metric to characterize the transmission reliability.
According to Eq. (\ref{esti sym}), the symbol MSE of $\bm{\bar{x}}_{\rm C}$, $\mathbb{E}\left(\Vert\tilde{\bm{x}}_C-\bm{\bar{x}}_{\rm C}\Vert_2^2\right)$, is derived as
    \begin{align}\label{MSE}
    &{\rm MSE_C}(\bm{H}_{\rm C},\bm{H}_{\rm S},\bm{P}_{\rm C},\bm{P}_{\rm S})\nonumber\\
    =&\frac{N_{\rm C}}{K}\Vert\bm{W}_{\rm BB}\bm{H}_{\rm C}\bm{P}_{\rm C}-\bm{I}_K\Vert_F^2\nonumber\\
    &+\Vert\bm{D}^\frac{1}{2}\bm{W}_{\rm BB}\bm{H}_{\rm S}\bm{P}_{\rm S}\Vert_F^2 +\sigma^2\Vert\bm{W}_{\rm BB}\Vert_F^2.
    \end{align}
As can be seen, the MSE stems from three factors: symbol estimation residual, sensing interference, and noise.
Since the sensing interference intensity may vary significantly with different channel realizations, it is difficult to determine a fixed MSE threshold remaining proper under all channel conditions. 
Therefore, we introduce a relative MSE threshold related to $\bm{H}_{\rm C}$ and $\bm{H}_{\rm S}$ as follows:

Assuming no processing is applied on the digital part, i.e., 
 $\bm{P}_{\rm S}={\rm{diag}}\left(\bm{t}\right)$ and $\bm{P}_{\rm C}=\bm{I}_K$, the corresponding digital combiner becomes
\begin{align}
\bm{W}_{\rm BB,0}=\frac{N_{\rm C}}{K}\bm{H}_{\rm C}^{\rm H}&\left(\frac{N_{\rm C}}{K}\bm{H}_{\rm C}\bm{H}_{\rm C}^{\rm H}\right.+\nonumber\\
&\bm{H}_{\rm S}\left({\rm{diag}}\left(\bm{t}\right)\right)^2\bm{D}\bm{H}_{\rm S}^{\rm H}+\sigma^2\bm{I}_K\left.\right)^{-1}.
\end{align}
The relative symbol MSE threshold is defined as
{\begin{align}\label{threshold}
\Gamma(\bm{H}_{\rm C},\bm{H}_{\rm S},\mu)=&\frac{N_{\rm C}}{K}\Vert\bm{W}_{\rm BB,0}\bm{H}_{\rm C}-\bm{I}_K\Vert_F^2+\nonumber\\
\mu\Vert\bm{D}^\frac{1}{2}\bm{W}_{\rm BB,0}&\bm{H}_{\rm S}{\rm{diag}}(\bm{t})\Vert_F^2+\sigma^2\Vert\bm{W}_{\rm BB,0}\Vert_F^2,
\end{align}}

{\noindent where $0 \le\mu\le 1$ is the weighting coefficient, signifying the relative tolerance for sensing interference cancellation errors. As $\mu$ increases, a lighter emphasis is placed on the communication side, resulting in improved sensing performance.}
\subsection{{Problem Formulation}}\label{original opt.}
\begin{table}[t]
\caption{{Description of each component in ISAC hybrid transceivers}}
\label{notation}
\centering
\begin{tabular}{c|c}
\hline
 \textbf{Symbol} & \textbf{Parameter} \\ \hline
$\bm{P}_{\rm C}$  & Digital communication precoder \\ \hline
$\bm{P}_{\rm S}$ & Digital sensing precoder \\ \hline
$\bm{F}_{\rm C}$ & Analog communication precoder \\ \hline
$\bm{F}_{\rm S}$ & Analog sensing precoder \\ \hline
$\bm{W}_{\rm RF}$ & Analog communication combiner \\ \hline
\end{tabular}
\end{table}
To minimize sensing beampattern MSE while adhering to constraints on symbol MSE, transmit power, as well as the analog precoder, a joint optimization problem of hybrid transceivers is formulated as
\begin{subequations}
\begin{align}
\min_{\substack{\bm{P}_{\rm C},\bm{P}_{\rm S}\\ \bm{F}_{\rm C},\bm{F}_{\rm S},\bm{W}_{\rm RF}}} &{\rm MSE_S}(\bm{P}_{\rm S},\bm{F}_{\rm S})\\
 s.t. \quad\enspace\, &{\rm MSE_C}(\bm{H}_{\rm C},\bm{H}_{\rm S},\bm{P}_{\rm C},\bm{P}_{\rm S}) \hspace{-0.04in}\le\hspace{-0.04in}\Gamma(\bm{H}_{\rm C},\bm{H}_{\rm S},\mu), \label{MSE Cons}\\
 &\Vert\bm{P}_{\rm C}\Vert_F^2\le K,\label{comm power}\\
 &\Vert\bm{D}^{\frac{1}{2}}\bm{b}\Vert_F^2\le T_{\rm R}, \label{radar power0}\\ &\bm{F}_{\rm C}\in\mathcal{F}, \label{ct analog set}
 \\ &\bm{F}_{\rm S}\in\mathcal{F}, \label{r analog set}\\
 &\bm{W}_{\rm RF}\in\mathcal{W}.\label{cr analog set}
\end{align}
\end{subequations}
The expressions of $\chi$ and $\Gamma(\mu)$ are provided in Eqs. (\ref{MSE}) and (\ref{threshold}), respectively. 
(\ref{MSE Cons}) represents MSE constraint on communication symbol estimation, (\ref{comm power}) corresponds to the communication power constraint, and (\ref{radar power0}) denotes the average sensing power constraint. 
{In Eqs. (\ref{ct analog set})-(\ref{cr analog set}), $\mathcal{F}$ and $\mathcal{W}$ represent the feasible analog precoder sets of the ISAC transmitter and communication receiver, respectively. We consider two representative analog configurations, i.e., MBS and MAS.} 
For MBS, {low-cost lens array antennas \cite{gao2018low} are employed} and each column of analog precoders is selected from $N$-dimensional DFT codebook $\mathcal{F}_{N}=\left\{\bm{F}_N[:,1],...,\bm{F}_N[:,N]\right\}$, where $N$ takes on $N_{\rm t}$ or $N_{\rm r}$.
For MAS, the fully connected $B$-bit phase shifter network is adopted, and the adjustable angles of each element of analog precoders are selected from $\mathcal{B}=\left\{0,\frac{2\pi}{2^B},...,\frac{2\pi(2^B-1)}{2^B}\right\}$.
{For convenience, the notations of all optimized variables are summarized in Table \ref{notation}.}
\section{{Hybrid Transceiver design for BPM-ISAC}}
In this section, we propose an efficient transceiver design to address the optimization problem formulated above. 
Due to the coupling of all five variables in the non-convex constraint (\ref{MSE Cons}) and the discrete nature of the feasible analog precoder set, the original problem is a complex mixed-integer non-convex large-scale combinatorial optimization problem that is difficult to solve. 
As a result, we address it by optimizing analog and digital parts sequentially.
\vspace{-3mm}
\subsection{{{Analog-part Optimization for BPM-ISAC}}}\label{section analog}
Considering that the analog part forms the EDC and plays a fundamental role in digital part design \cite{gao2021mutual}, we first optimize the analog part with the unoptimized digital part. 
Then the symbol MSE is converted to the function of 
$\bm{H}_{\rm S}$ and $\bm{H}_{\rm C}$, i.e.,
\begin{align}
    \overline{\rm MSE}_{\rm C}(\bm{H}_{\rm C},\bm{H}_{\rm S})=&\frac{N_{\rm C}}{K}\Vert\bm{W}_{\rm BB,0}\bm{H}_{\rm C}-\bm{I}_K\Vert_F^2+\nonumber\\
    \Vert\bm{D}^\frac{1}{2}\bm{W}_{\rm BB,0}&\bm{H}_{\rm S}{\rm{diag}}(\bm{t})\Vert_F^2+\sigma^2\Vert\bm{W}_{\rm BB,0}\Vert_F^2.
\end{align}

It is worth noting that $\overline{\rm MSE}_{\rm C}(\bm{H}_{\rm C},\bm{H}_{\rm S})$ is the upper bound of the relative symbol MSE threshold, which measures the communication performance of EDC.

Firstly, the analog sensing precoder $\bm{F}_{\rm S}$ is optimized to point in the direction of interest by solving the following problem.
\begin{subequations}
\begin{empheq}[box=\fbox]{align}
	\mathcal{P}.1\mbox{-}1:~&\textbf{{Optimization for analog sensing precoder}}\nonumber\\
&\mathop{\min}_{\bm{F}_{\rm S}} \enspace{\rm MSE_S}(\bm{F}_{\rm S})\nonumber\\
& \ s.t.  \enspace \ \bm{F}_{\rm S}\in\mathcal{F}.\nonumber
\end{empheq}
\end{subequations}
Then $\bm{F}_{\rm C}$ and $\bm{W}_{\rm RF}$ are jointly optimized with fixed $\bm{F}_{\rm S}$ to minimize symbol MSE as follows:
\noindent
\begin{subequations}
\begin{empheq}[box=\fbox]{align}
\mathcal{P}.1\mbox{-}2:~&\textbf{\hspace{-1mm}Optimization for communication's analog part}\nonumber\\
&\mathop{\min}_{\bm{F}_{\rm C},\bm{W}_{\rm RF}}\overline{\rm MSE}_{\rm C}(\bm{H}_{\rm C},\bm{H}_{\rm S})\nonumber\\
& \quad s.t. \quad \bm{F}_{\rm C}\in\mathcal{F},\nonumber\\
&\qquad \quad \  \bm{W}_{\rm RF}\in\mathcal{W}.\nonumber
\end{empheq}
\end{subequations}
Below, $\mathcal{P}.1\mbox{-}1$ and $\mathcal{P}.1\mbox{-}2$ are solved with the configuration of MBS and MAS, respectively.
\subsubsection{Analog-part for {MBS}}
{Firstly, the following proposition illustrates that  $\mathcal{P}.1\mbox{-}1$ can be transformed into a series of parallel optimization problems.
\begin{proposition}\label{parallel}
Solving $\mathcal{P}.1\mbox{-}1$ is equivalent to optimizing each column of $\bm{F}_{\rm S}$ individually as follows:
\begin{subequations}
    \begin{align}
&\mathop{\max}_{\bm{F}_{\rm S}[:,l]}\lvert\bm{a}_{N_{\rm t}}^{\rm H}(\theta_l)\bm{F}_{\rm S}[:,l]\rvert\nonumber\\
&\ s.t. \quad \bm{F}_{\rm S}[:,l]\in \mathcal{F}_{N_{\rm t}}.\nonumber
\end{align}
\end{subequations}
\end{proposition}
\begin{proof}
Under an unoptimized digital part, $\bm{P}_{\rm S}={\rm{diag}}\left(\bm{t}\right)$. Then the objective function can be rewritten as
\begin{align}\label{MSE_S}
{\rm MSE_S}=\sum_{l=1}^{W}d_lt_l(\lvert\bm{a}_{N_{\rm t}}^{\rm H}(\theta_l)\bm{F}_{\rm S}[:,l]\rvert-1)^2.
\end{align}
Considering that $0\leq\lvert\bm{a}_{N_{\rm t}}^{\rm H}(\theta_l)\bm{F}_{\rm S}[:,l]\rvert \leq 1$, minimizing Eq. (\ref{MSE_S}) is equivalent to maximizing $\lvert\bm{a}_{N_{\rm t}}^{\rm H}(\theta_l)\bm{F}_{\rm S}[:,l]\rvert$ parallelly.
\end{proof}}
\vspace{-3mm}
The optimal solution is obtained by exhaustive search and the computational complexity is $\mathcal{O}(WN_{\rm t})$, which is {obviously} acceptable.
Considering the sensing beam set as $\Omega=\{\bm{F}_{\rm S}[:,1],...,\bm{F}_{\rm S}[:,W]\}$,  communication transmitting beams should be selected from {set difference} $\left\{\mathcal{F}_{N_{\rm t}}\backslash
\Omega\right\}$ to avoid sensing interference on the communication receiver.

For $\mathcal{P}.1\mbox{-}2$, the optimization problem is reformulated as
\begin{subequations}
    \begin{align} \mathop{\min}_{\bm{F}_{\rm C},\bm{W}_{\rm RF}}&\overline{\rm MSE}_{\rm C}(\bm{H}_{\rm C},\bm{H}_{\rm S})\nonumber\\
s.t. \quad &\bm{F}_{\rm C}[:,i]\in\left\{\mathcal{F}_{N_{\rm t}}\backslash
\Omega\right\},\nonumber\\
&\bm{W}_{\rm RF}[:,i]\in\mathcal{F}_{N_{\rm r}}.\nonumber
    \end{align}
\end{subequations}
Similar to \cite{gao2019spatial}, $\bm{H}_{\rm C}$ and $\bm{H}_{\rm S}$ are the submatrices of beamspace channel $\bm{\bar{H}}=\bm{F}_{N_{\rm r}}^H\bm{H}\bm{F}_{N_{\rm t}}$.
Specifically, the indices of the selected sub-columns and sub-rows indicate the selected DFT codewords for transmitted and received beams respectively.
Then the above problem can be regarded as determining a $K\times K$ submatrix of $\bm{\bar{H}}$.

Considering the exponential time complexity of the exhaustive search, we proposed a two-stage alternative to lower complexity.
Firstly, a set of $L$ largest elements in $\bm{\bar{H}}$ are selected as candidate beam pairs.
Secondly, the final columns and rows are chosen within the scope of these $L$ candidate beam pairs using the min-MSE criterion. 
Then the overall computational complexity is reduced from $\mathcal{O}(C_{N_{\rm r}}^K C_{N_{\rm t}-W}^K)$ to $\mathcal{O}(N_r^2(N_t-W)^2+C_L^K)$.

\subsubsection{Analog-part for {MAS}}
{Compared to MBS, MAS has higher degrees of freedom for more flexible beam patterns and better EDC.
Below, we resort to the MAS with $B$-bit PSs for analog part design.}
In this case,  $\mathcal{P}.1\mbox{-}1$ and $\mathcal{P}.1\mbox{-}2$ are integer programming problems, and the optimal solution can be solved by brutal search. 
However, it is impractical due to the exponential growth of time complexity with both the number of bits and antennas. 
Thus low-complexity methods are proposed for these two problems.

For $\mathcal{P}.1\mbox{-}1$, according to proposition \ref{parallel}, it is equivalent to solving the following problems parallelly.
\noindent
\begin{subequations}
    \begin{align}
&\mathop{\max}_{\bm{F}_{\rm S}[:,l]}\lvert\bm{a}_{N_{\rm t}}^{\rm H}(\theta_l)\bm{F}_{\rm S}[:,l]\rvert\nonumber\\
&\ s.t. \quad \bm{F}_{\rm S}[i,l]\in\frac{e^{j\mathcal{B}}}{\sqrt{N_{\rm t}}}, \forall i.\nonumber
\end{align}
\end{subequations}
{Since the definition domain of the variable is finite, it can be solved by the widely-used branch and bound (B$\&$B) algorithm \cite{morrison2016branch}, which adopts tree search strategy but applies pruning rules to skip suboptimal regions of the tree.
The tree has $N_{\rm t}+1$ levels, and the $n\mbox{-}$th level branch represents the value of $n\mbox{-}$th PSs, which has $2^B$ child nodes.
We perform a breadth-first search with $N_{\rm t}$ iterations, during which it maintains the currently best available solution $\bm{q}^*$, global lower bound $\mathcal{L}$ and available set $\mathcal{G}$. 
The overall algorithm is summarized in Algorithm 1.
At initialization, $\mathcal{G}$ is consist of the root note $v_0$,  $\bm{q}^*$ is randomly set as $\bm{q}_0$ and $\mathcal{L}$ is initialized as $\mathcal{L}_0=\lvert\bm{a}_{N_{\rm t}}^{\rm H}(\theta_l)\bm{q}_0\rvert$.
In the $n\mbox{-}$th iteration, for each new node, its upper bound and lower bound are calculated for updates.}

For each node $v$, denote the first $n$ elements of $\bm{F}_{\rm S}[:,l]$ as $\bm{q}_{\rm L}\in \mathbb{C}^{n\times 1}$ and the remained $N_{\rm t}-n$ elements as $\bm{q}_{\rm R}\in \mathbb{C}^{(N_{\rm t}-n)\times 1}$.
The objective function can be rewritten as $f=\left\lvert\bm{a}_{N_{\rm t}}^{\rm H}(\theta_l)[1:n]\bm{q}_{\rm L}+\bm{a}_{N_{\rm t}}^{\rm H}(\theta_l)[n+1:N_{\rm t}]\bm{q}_{\rm R}\right\rvert$.
According to triangle inequality, the upper bound can be achieved as 
\begin{align}\label{f_UB}
f_{\rm UB}=\lvert\bm{a}_{N_{\rm t}}^{\rm H}(\theta_l)[1:n]\bm{q}_{\rm L}\rvert+N_{\rm t}-n
\end{align}
with $\bm{q}_{\rm R}=\bm{q}_{\rm R,UB}$.
The $i\mbox{-}$th phase of  $\bm{q}_{\rm R,UB}$ satisfies
\begin{align}
\angle\bm{q}_{\rm R,UB}[i]=\angle(\bm{a}_{N_{\rm t}}^{\rm H}(\theta_l)[1:n]\bm{q}_{\rm L})+\angle(\bm{a}_{N_{\rm t}}(\theta_l)[i]),
\end{align}
where $\angle(x)$ represents the phase of $x$.
If $f_{\rm UB}$ is lower than the current global lower bound $\mathcal{L}$, all leaves below node 
$v$ are suboptimal and will be pruned.
{Otherwise, feasible $\bm{q}_{\rm R}$ will be obtained as $\bm{q}_{\rm R, LB}$ by quantifying $\bm{q}_{\rm R,UB}$ nearby with given $B$ bit resolution, i.e.,
\begin{align}
\bm{q}_{\rm R,LB}=\mathop{\arg\min}\limits_{\bm{q}_{R}[i]\in \frac{e^{j\mathcal{B}}}{\sqrt{N_{\rm t}}},\forall i}\Vert\bm{q}_{\rm R}-\bm{q}_{\rm R,UB}\Vert_2.
\end{align}}
Meanwhile, the lower bound can be obtained as
\begin{align}\label{f_LB}
f_{\rm LB}=\left\lvert\bm{a}_{N_{\rm t}}^{\rm H}(\theta_l)[1:n]\bm{q}_{\rm L}+\bm{a}_{N_{\rm t}}^{\rm H}(\theta_l)[n+1:N_{\rm t}]\bm{q}_{\rm R,LB}\right\rvert.
\end{align}
If $f_{\rm LB}$ is larger than $\mathcal{L}$, $\mathcal{L}$ will be updated as $f_{\rm LB}$ and $\bm{q}^*$ should be updated as $[\bm{q}_{\rm L}^{\rm T}, \bm{q}_{\rm R,LB}^{\rm T}]^{\rm T}$. 
It is noteworthy that the solution of B$\&$B algorithm is optimal with reduced computation time in comparison with the exhaustive search.
\begin{algorithm}[t]
\caption{{Branch and Bound Algorithm for Optimizing Analog Sensing Precoder}}
\label{B&B}
\hspace*{0.02in} {\bf Input:}
$\bm{a}(\theta_l)$.
\begin{algorithmic}[1]
\State \textbf{Initialization}: $\mathcal{G}=\{v_0\}$, $\bm{q}^*=\bm{q}_0$ and $\mathcal{L}=\mathcal{L}_0$.
\For {$n=1:N_{\rm t}$}
\State Replace each node in $\mathcal{G}$ with its descendant nodes.
\For{$v\in\mathcal{G}$}
\State Calculate $f_{\rm UB}$ and $f_{\rm LB}$ by Eq. (\ref{f_UB}) and Eq. (\ref{f_LB}).
\State Update $\mathcal{G}$ as $\mathcal{G}\backslash \{v\}$ if $f_{\rm UB}\le \mathcal{L}$.
\State Update $\mathcal{L}$ and $\bm{q}^*$ if $f_{\rm LB}>\mathcal{L}$.
\EndFor
\EndFor
\end{algorithmic}
\hspace*{0.02in} {\bf Output:} $\bm{F}_{\rm S}[:,l]=\bm{q}^*$.
\end{algorithm}
For $\mathcal{P}.1\mbox{-}2$, it has been shown \cite{wang2018hybrid} that entry-wise iteration can be a low-complexity effective method for finite resolution PSs case. In detail, for columns $z$ from 1 to $K$, the $\bm{F}_{\rm C}[:,z]$ and $\bm{W}_{\rm RF}[:,z]$ are optimized successively.
For each column, each entry is optimized to minimize objective function while keeping the others fixed until convergence. 
For example, $\bm{F}_{\rm C}[i,z]$ can be updated by solving
\noindent
\begin{subequations}
    \begin{align} \mathop{\min}_{\bm{F}_{\rm C}[i,z]}&\overline{\rm MSE}_{\rm C}\nonumber\\
s.t. \ \  &\bm{F}_{\rm C}[i,z]\in\frac{e^{j\mathcal{B}}}{\sqrt{N_{\rm t}}}.\nonumber
\end{align}
\end{subequations}

\vspace{-1mm}
Considering that appropriate initialization can improve the performance of the proposed iteration algorithm \cite{yu2022dynamic}, an improved initializer is proposed below, especially for the low-bit case. 
Firstly, the following proposition illustrates the transformation of the initial problem.
\vspace{-2mm}
\begin{proposition}\label{MISDP}
{As the SNR increases, $\mathcal{P}.1\mbox{-}2$ with $B$-bit PSs tends to be asymptotically equivalent to the following mixed-integer semi-definite programming (MISDP) problem:}
\noindent
\begin{subequations}
\begin{align}
\mathop{\min}_{w,\bm{F}_{\rm C},\bm{W}_{\rm RF}}&w\\
s.t.\ \quad &\bm{Z}=
\begin{bmatrix}
\frac{w}{K+W}\bm{I}_{K+W} & \bm{G} \\ \bm{G}^{\rm H} & \bm{J}
\end{bmatrix}
\succeq 0,\label{SDP constraint} \\ 
&\bm{F}_{\rm C}[i,j]\in\frac{e^{j\mathcal{B}}}{\sqrt{N_{\rm t}}},\forall i,j,\\
&\bm{W}_{\rm RF}[i,j]\in\frac{e^{j\mathcal{B}}}{\sqrt{N_{\rm t}}},\forall i,j,
\end{align}
\end{subequations}
where
\vspace{-1mm}
\begin{align}
\bm{G}=
\begin{bmatrix}
{\rm{diag}}(\bm{t})\bm{D}^{\frac{1}{2}}\bm{F}_{\rm S}\bm{H}^{\rm H}\bm{W}_{\rm RF}\nonumber\\
\sigma \bm{I}_K
\end{bmatrix},
\end{align}
and 
\vspace{-2mm}
\begin{align}
\bm{J}=\bm{W}_{\rm RF}^{\rm H}\bm{H}\bm{F}_{\rm C}\bm{F}_{\rm C}^{\rm H}\bm{H}^{\rm H}\bm{W}_{\rm RF}.\nonumber
\end{align}
\end{proposition}
\begin{proof}
    See Appendix \ref{A}.
\end{proof}
\vspace{-2mm}
Considering that $\bm{F}_{\rm C}$ and $\bm{W}_{\rm RF}$ are coupled in the nonlinear constraint (\ref{SDP constraint}), we solve the problem by alternatively optimizing $\bm{F}_{\rm C}$ and $\bm{W}_{\rm RF}$ until convergence.
Below, we take $\bm{F}_{\rm C}$ as an example to illustrate how to transform the constraint (\ref{SDP constraint}) into a linear matrix inequality (LMI) constraint.

Denote $\bm{a}=\frac{\pi}{2^{B-1}}[-(2^B-1),...,(2^B-1)]^{\rm T}$, $\bm{c}=\cos(\bm{a})$ and $\bm{s}=\sin(\bm{a})$. 
A series of binary vectors $\bm{x}^{i,j}\in\mathbb{C}^{(2^{B+1}-1)\times1}$ and $\bm{y}^{i,j,i^{\prime},j^{\prime}}\in\mathbb{C}^{(2^{B+1}-1)\times1}$ are introduced, where $\bm{x}^{i,j}[t]$ indicates whether $\angle\bm{F}_{\rm C}[i,j]$ is $\bm{a}[t]$ and $\bm{y}^{i,j,i^{\prime},j^{\prime}}[t]$ indicates whether $\angle\bm{F}_{\rm C}[i,j]-\angle\bm{F}_{\rm C}[i^{\prime},j^{\prime}]$ is $\bm{a}[t]$. 
\begin{lemma}
The problem in proposition \ref{MISDP} with fixed $\bm{W}_{\rm RF}$ can be equivalently transformed as 
\begin{align}
\mathop{\min}_{w,\bm{x}^{i,j},\bm{y}^{i,j,i^{\prime},j^{\prime}}}&w\nonumber\\
s.t.\quad \quad & (\ref{SDP constraint}), \Vert\bm{x}^{i,j}\Vert_1=1, \bm{e}^{\rm T}\bm{x}^{i,j}=0,\nonumber\\
&\bm{a}^{\rm T}(\bm{x}^{i,j}-\bm{x}^{i^{\prime},j^{\prime}})=\bm{a}^{\rm T}\bm{y}^{i,j,i^{\prime},j^{\prime}},\nonumber
\end{align}
where $\bm{e}\in\mathbb{C}^{(2^{B+1}-1)\times1}$ is a constant vector whose first $2^B-1$ elements are 1 and others are 0. $\bm{J}$ is the linear function of  $\bm{y}^{i,j,i^{\prime},j^{\prime}}$.
\end{lemma}
\begin{proof}
    See Appendix \ref{LMI}.
\end{proof}
\vspace{-2mm}
Similarly, for the solution of  $\bm{W}_{\rm RF}$, we can transform
$\bm{G}$ and $\bm{J}$ to the  linear function of $\bm{x}^{i,j}$ and $\bm{y}^{i,j,i^{\prime},j^{\prime}}$.
The above problem is a MISDP problem with linear constraints and can be solved using the outer approximation method, which can seek existing optimization toolboxes such as YALMIP \cite{lofberg2004yalmip}. 
The algorithm is summarized as Algorithm 2.
\begin{algorithm}[t]
\caption{Entry-wise Iteration Algorithm with MISDP Initialization for Optimizing Communication's Analog Part}
\label{alg MISDP}
\hspace*{0.02in} {\bf Input:}
$\bm{H}$, $\bm{F}_{\rm S}$.
\begin{algorithmic}[1]
\State \textbf{Initialization}: $\bm{F}_{\rm C}=\bm{1}_{N_{\rm t}\times K}$,$\bm{W}_{\rm RF}=\bm{1}_{N_{\rm r}\times K}$.
\State Alternatively optimize $\bm{F}_{\rm C}$ and $\bm{W}_{\rm RF}$ by solving MISDP problem until convergence to obtain an initial value
\For {z=1:$K$}
\Repeat
\For {$i=1:N_{\rm r}$}
\State Optimize $\bm{F}_{\rm C}[i,z]$ to minimize $\bar{\chi}$.
\EndFor
\For {$i=1:N_{\rm t}$}
\State Optimize $\bm{W}_{\rm RF}[i,z]$ to minimize $\bar{\chi}$.
\EndFor
\Until Convergence.
\EndFor
\end{algorithmic}
\hspace*{0.02in} {\bf Output:} $\bm{F}_{\rm C}$,$\bm{W}_{\rm RF}$.
\end{algorithm}
\vspace{-3mm}
\subsection{{{Digital-part Optimization for BPM-ISAC}}}
\vspace{-1mm}
With the fixed analog part, the original optimization problem in Section \ref{original opt.} is transformed as follows:
\begin{subequations}\label{digital}
\begin{empheq}[box=\fbox]{align}
\mathcal{P}.2:~&\textbf{Optimization for communication and}\nonumber\\
&\quad \textbf{sensing's digital part}\nonumber\\
\min_{\substack{\bm{P}_{\rm C},\bm{P}_{\rm S}}} \enspace & {\rm MSE_S}(\bm{P}_{\rm S}) \label{digital opject}\nonumber\\
 s.t.\ \  & \left(\ref{MSE Cons}\right)-\left(\ref{radar power0}\right).\nonumber
\end{empheq}
\end{subequations}
{Since $\bm{P}_{\rm S}$, $\bm{P}_{\rm C}$ and $\bm{W}_{\rm BB}$ are coupled in Eq. (\ref{m}) in a non-convex manner, $\mathcal{P}.2$ is still non-convex.
Thus an alternating optimization algorithm is proposed to alternatively update \{$\bm{P}_{\rm S}$, $\bm{P}_{\rm C}$\} and $\bm{W}_{\rm BB}$}. For initialization, we set $\bm{P}_{\rm S}$, $\bm{P}_{\rm C}$, and $\bm{W}_{\rm BB}$ as ${\rm diag}(\bm{t})$, $\bm{I}_K$, and $\bm{W}_{\rm BB,0}$, respectively. For each iteration, the following two steps are executed in order.\\
{1) Update $\bm{P}_{\rm S}$ and $\bm{P}_{\rm C}$ with fixed $\bm{W}_{\rm BB}$. $\mathcal{P}.2$ can be recast into the following problem with $\bm{p}$ and $\bm{b}$ as variables.
\begin{equation}
    \begin{aligned}
\min_{\bm{b},\bm{p}} \enspace &\sum_{i=1}^W d_i(\lvert\bm{a}_{N_{\rm t}}^{\rm H}(\theta_i)\bm{F}_{\rm S}[:,i]\rvert b_i-t_i)^2\label{QCQP}\\
      s.t. \ \  &\Vert{\rm{diag}}(\bm{p})\Vert_F^2\le K,\\
      &\Vert\bm{D}^{\frac{1}{2}}\bm{b}\Vert_F^2\le T_{\rm R}, \\
      & \frac{N_{\rm C}}{K}\Vert\bm{W}_{\rm BB}\bm{H}_{\rm C}{\rm{diag}}(\bm{p})-\bm{I}_K\Vert_F^2\\
      & +\Vert\bm{D}^\frac{1}{2}\bm{W}_{\rm BB}\bm{H}_{\rm S}{\rm{diag}}(\bm{b})\Vert_F^2 +\sigma^2\Vert\bm{W}_{\rm BB}\Vert_F^2 \le \Gamma.
    \end{aligned}
\end{equation}}
Since both the objective function and constraints are quadratic, the above problem is a convex quadratically constrained quadratic program (QCQP) problem, which can be solved via convex optimization toolbox \cite{lofberg2004yalmip}.\\
2) Update $\bm{W}_{\rm BB}$ with fixed $\bm{P}_{\rm S}$ and $\bm{P}_{\rm C}$ according to Eq. (\ref{m}).
The overall alternating algorithm is presented in Algorithm \ref{alg digital}.

\begin{algorithm}[t]
\caption{Alternating Optimization Algorithm for Optimizing Communication and Sensing's Digital Part}
\label{alg digital}
\hspace*{0.02in} {\bf Input:} 
$\bm{H}_{\rm C}$, $\bm{H}_{\rm S}$, $N_{\rm C}$, $K$, $\sigma$.
\begin{algorithmic}[1]
\State Set $\bm{P}_{\rm S}^{(0)}={\rm diag}(\bm{t})$, $\bm{P}_{\rm C}^{(0)}=\bm{I}_K$ and $\bm{W}_{\rm BB}^{(0)}=\bm{W}_{\rm BB,0}$.
\Repeat
\State {Update $\bm{P}_{\rm S}^{(i+1)}$ and $\bm{P}_{\rm C}^{(i)}$ with fixed $\bm{W}_{\rm BB}^{(i)}$ by solving QCQP problem (\ref{QCQP}).}
\State Update $\bm{W}_{\rm BB}^{(i+1)}$ with $\bm{P}_{\rm S}^{(i+1)}$ and $\bm{P}_{\rm C}^{(i+1)}$ as Eq. (\ref{m}).
\Until the value of the objective function converges.
\end{algorithmic}
\hspace*{0.02in} {\bf Output:} $\bm{P}_{\rm S}$, $\bm{P}_{\rm C}$, $\bm{W}_{\rm BB}$.
\end{algorithm}
\section{{ISAC Performance Analysis}}
{In this section, the complexity and convergence performance of the proposed algorithm are theoretically analyzed.
The APEP and CRB are derived to illustrate the theoretical performance of communication and sensing.
In addition, the number of RF chains for sensing are discussed.}
\vspace{-3mm}
\subsection{Complexity Analysis}
The complexity of the analog part design with MBS structures has been derived in Section \ref{section analog}.
For algorithm 1, the worst-case theoretical complexity of the B$\&$B algorithm is $\mathcal{O}(2^{BN_t})$, but the pruning rules can substantially reduce actual solving time.
For algorithm 2, the overall complexity includes the initialization process and entry-wise iteration. 
For initialization,  the problem in proposition 2 can be transformed into mix-integer linear programming problems and solved by the branch and bound algorithm, the complexity of which is $\mathcal{O}(2^{BN_{\rm t}K})$.
Since the initialization scheme is only applied to low-bit cases, the complexity is acceptable.
The complexity of the entry-wise iteration part is $\mathcal{O}(N_{\rm iter}K(N_{\rm t}+N_{\rm r})2^B)$, where $N_{\rm iter}$ denotes the number of iterations.
For algorithm 3, the complexity of solving QCQP problems is
$\mathcal{O}(N_{\rm iter}^{'}(T^{3.5}+W^{3.5}){\rm log}(1/\epsilon))$ \cite{luo2010semidefinite} by the interior-point method given accuracy level $\epsilon$, where $N_{\rm iter}^{'}$ is the number of iteration rounds.
\subsection{Convergence Analysis}
Algorithm \ref{B&B} has a finite number of operational steps.
Algorithm \ref{alg MISDP} converges because the objective function is non-increasing and is lower-bounded by 0.
For Algorithm \ref{alg digital}, the convergence and existence of the solution are not obvious and analyzed as below.
For the first iteration, it can be observed that $\bm{b}^{(1)}=\mu \bm{t}$ and $\bm{p}^{(1)}=\bm{p}^{(0)}$ are the feasible solution for the first step.
Therefore, a solution must exist in the first iteration.
Suppose after $i\mbox{-}$th iteration, all constraints are satisfied. 
During $(i+1)\mbox{-}$th iteration, denote the objective value at step $j$ as $\varepsilon_j$. For step 1), we have
\begin{align}
\varepsilon_1(\bm{b}^{(i+1)})\le\varepsilon_1(\bm{b}^{(i)}),
\end{align}
 and all constraints except for Eq. (\ref{m}) are satisfied. 
 After step 2), Eq. (\ref{m}) is satisfied and 
  \begin{align}
\varepsilon_3(\bm{b}^{(i+1)})=\varepsilon_2(\bm{b}^{(i+1)}).
  \end{align} It is worth noting that the constraint (\ref{MSE Cons}) is still satisfied since $\bm{W}_{\rm BB}^{(i+1)}$ is the LMMSE equalizer, which further lowers the symbol MSE.
 Therefore, after $(i+1)\mbox{-}$th iteration, we have
 \begin{align}
     \varepsilon_3(\bm{b}^{(i+1)}) \le\varepsilon_1(\bm{b}^{(i)}),
 \end{align}
 i.e., the objective function is non-increasing and all constraints are satisfied. 
Recalling that the objective value is lower bounded, the convergence of the proposed alternating optimization is guaranteed.
{For the optimality of the convergence point, we have the following proposition.
\begin{proposition}
It can be proven that the convergence point of Algorithm \ref{alg digital} satisfies the Karush-Kuhn-Tucher (KKT) conditions of $\mathcal{P}.2$.
\end{proposition}
\begin{proof}
	See Appendix \ref{KKT_proof}.
\end{proof}}
In Fig. \ref{Convergence}, we set the convergence tolerance as 0.001 and the convergence performance of Algorithm 3 with different $\mu$ and $N_t$ is presented.
It can be observed that the convergence speed slows down as the value of $\mu$ decreases and the number of transmit antennas increases.
\begin{figure}[tbp]
    \centering
    \includegraphics[width=8cm]{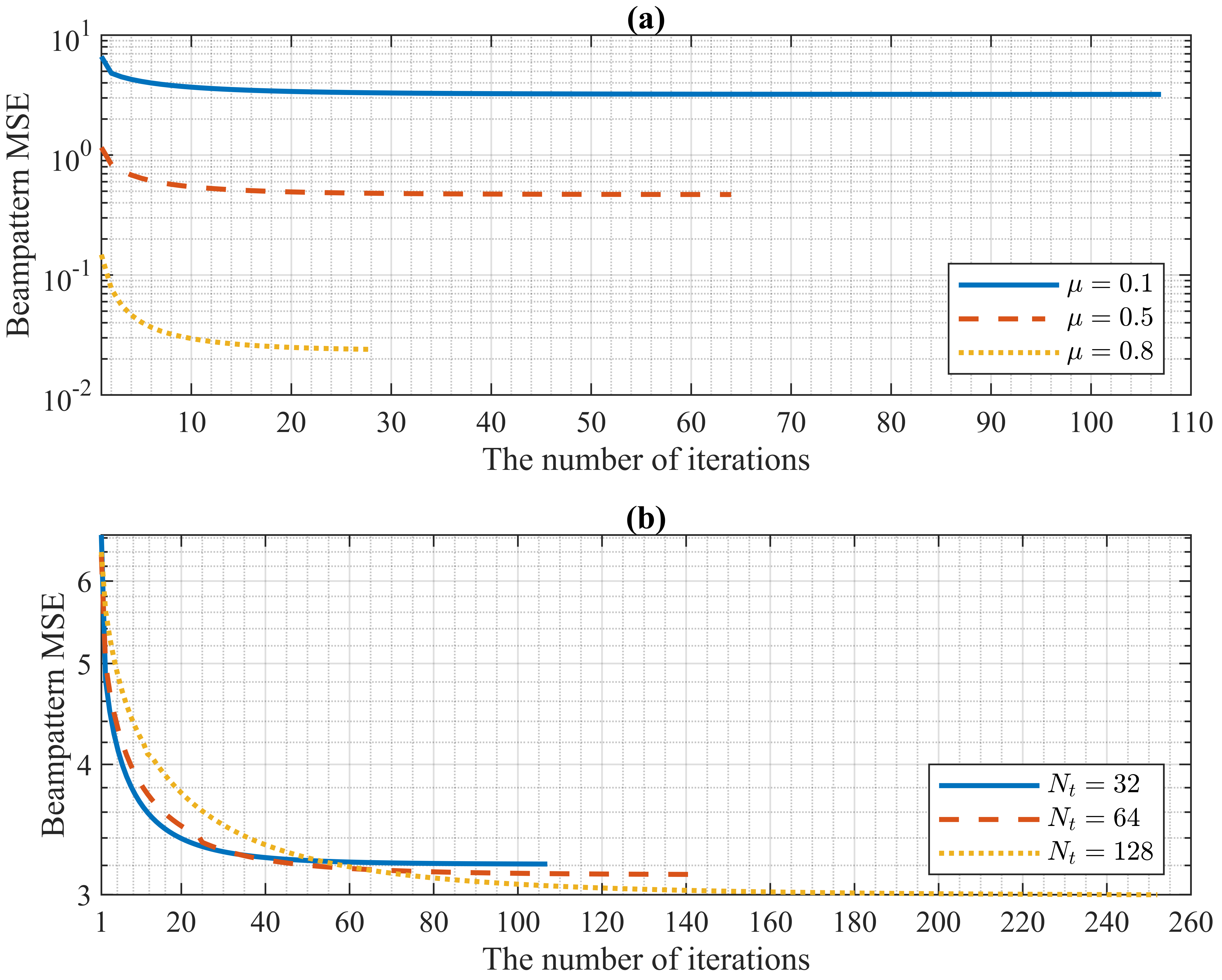}
    \vspace{-3mm}
    \caption{Algorithm \ref{alg digital}'s convergence behaviour with $N_{\rm r}=32$. (a): $N_{\rm t}=32$, $\mu=0.1,0.5,0.8$; (b): $N_{\rm t}=32,64,128$, $\mu=0.1$.}
    \label{Convergence}
    \vspace{-2mm}
\end{figure}
\vspace{-3mm}
\subsection{APEP Analysis}
The APEP is derived to illustrate the theoretical BER performance of the proposed scheme. 
Due to the presence of finite-bit PSs and digital part optimization, obtaining an exact APEP is challenging. 
For simplicity, we analyze sub-beamspace with MBS and the unoptimized digital part, assuming infinite sensing interference power. 

Firstly, we explain how the number of effective paths decreases due to the interference of sensing beams.
As shown in Fig. \ref{beamspace}, there are $P=7$ paths in the original $N_{\rm r} \times N_{\rm t}=7\times7$ beamspace channel and we neglect the off-grid beam leakage.
{BPM refers to the communication-only version of the proposed approach.}
For BPM-ISAC, $W$ sensing beams cover $M_{\rm R}=3$ paths, which further cover $M_{\rm B}=2$ received beams. 
For example, the path $\bm{\bar{H}}(1,2)$ cannot be used for communication because its received beam will be interfered with by $\bm{\bar{H}}(1,5)$. 
Therefore, the communication paths can only be chosen from the rest unaffected $(N_{\rm t}-W)\times(N_{\rm r}-M_{\rm B})$ beam pairs. 
In this case, the number of effective paths is $M_{\rm C}=3$.
In proposition \ref{probability distribution}, the probability distribution of the number of effective paths is derived.
\begin{figure}[tbp]
    \centering
    \includegraphics[width=9cm]{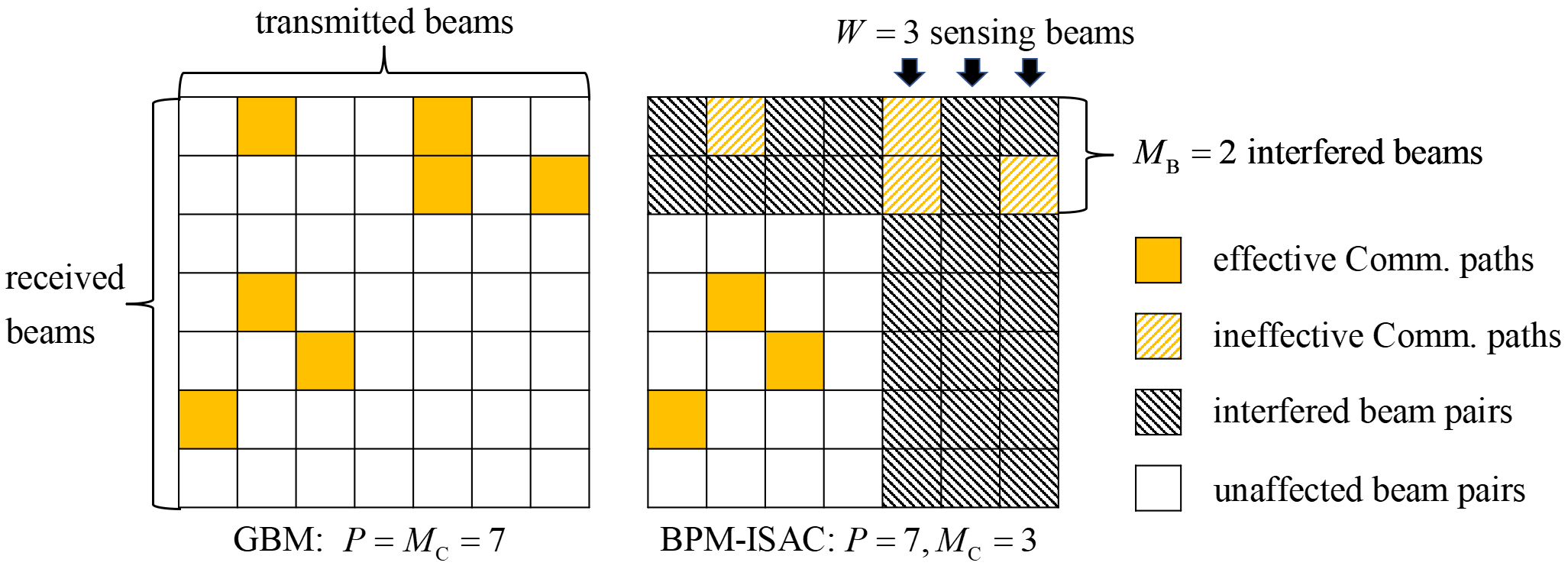}
    \caption{An illustration of beamspace channel for {BPM} and BPM-ISAC.}
    \label{beamspace}
\end{figure}
\begin{proposition}\label{probability distribution}
For the case where the mmWave channel contains $P$ paths and there are $W$ sensing beams, the probability distribution of the number of effective communication paths is
\begin{align}\label{P com}
P(M_{\rm C}=c)=
\left\{\begin{array}{l} P_{M_{\rm R}}(0), c=P
\\ \sum\limits_{r=1}^{P-c} \sum\limits_{b=1}^{r}P_{M_{\rm R}}(r)P_{M_{\rm B}}(r,b)P_{M_{\rm C}}(c,r,b), 
\\\qquad \qquad \qquad \qquad c=0,...,P-1.
\end{array}\right.
\end{align}
where 
\begin{align}\label{P_MR}
 P_{M_{\rm R}}(r)=\frac{C_{N_{\rm r}(N_{\rm t}-W)}^{P-r}C_{N_{\rm r}W}^{r}}{C_{N_{\rm t}N_{\rm r}}^P}, 
\end{align}
\begin{align}\label{P_MB}
P_{M_{\rm B}}(r,b)=
\left\{\begin{array}{l}P_{M_{\rm B}}(r-1,b-1)\frac{(N_{\rm r}-b+1)W}{N_{\rm r}W-r+1}\\
\qquad +P_{M_{\rm B}}(r-1,b)\frac{bW-r+1}{N_{\rm r}W-r+1}, o.w.
\\0, \quad \qquad (r,b)=(1,1) \quad or \quad b=0.
\end{array}\right.
\end{align}
\begin{align}\label{P_MC}
P_{M_{\rm C}}(c, r, b)=\frac{C_{b(N_{\rm t}-W)}^{P-r-c}C_{(N_{\rm r}-r)(N_{\rm t}-W)}^c}{C_{N_{\rm r}(N_{\rm t}-W)}^{P-r}}.
\end{align}
\end{proposition}
\begin{proof}
    See Appendix \ref{B}.
\end{proof}
\begin{figure}[tbp]
    \centering
\includegraphics[width=8cm]{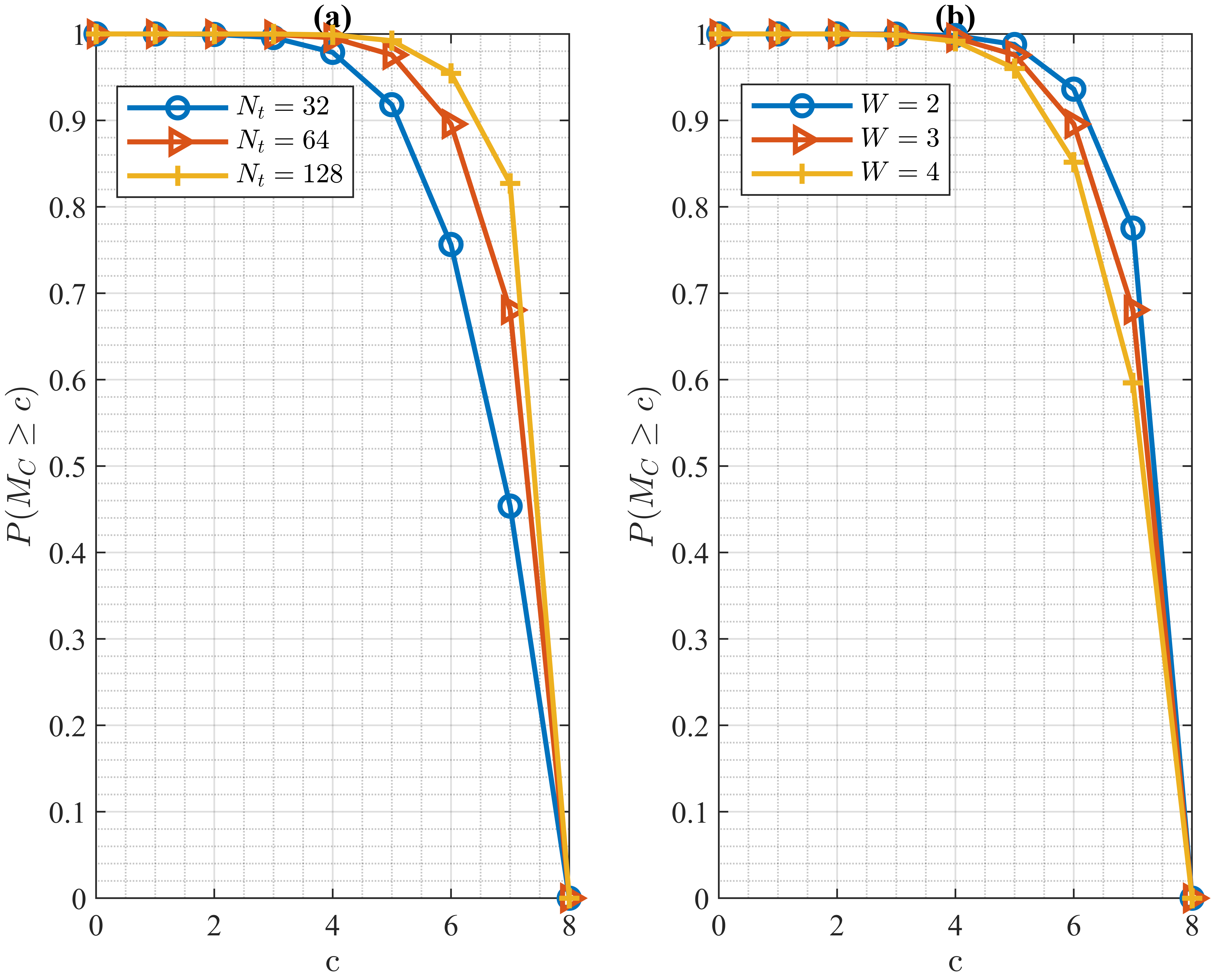}
\vspace{-2mm}
    \caption{The probability distribution of the number of effective paths with $N_{\rm r}=32$. (a): $W=3$; (b): $N_t=32$.}
    \label{distribution}
    \vspace{-5mm}
\end{figure}
In Fig. \ref{distribution}, the probability of the number of effective paths $M_C\geq c$ is given. 
It can be observed that a larger number of transmit antennas and fewer sensing beams will render more effective communication beams.
\begin{proposition}\label{pairwise proposition}
For the case where the mmWave channel contains $P$ paths and there are $W$ sensing beams, the pairwise error probability $P(\bm{\bar{x}}_{\rm C} \rightarrow \bm{\hat{x}}_C)$ through maximum-likelihood (ML) detection algorithm is derived as Eq. (\ref{pairwise}).
\begin{figure*}[hb]
\vspace{-3mm}
\hrulefill
\setcounter{TempEqCnt}{\value{equation}} 
\setcounter{equation}{43} 
\begin{align}\label{pairwise}
 P(\bm{\bar{x}}_{\rm C}\rightarrow \bm{\hat{x}}_C)
\simeq \sum_{c=K}^{P}\frac{c!P\left(M_{\rm C}=c\right)}{\left(c-K\right)!}&
\left(\frac{\mathbb{B}\left(\sum_{i=1}^{K}\left(\frac{N_{\rm t}N_{\rm r}}{4P\sigma^2}{\bigtriangleup x_i^2}+1\right),c-K+1\right)}{12\prod \limits_{j=2}^{K}\sum_{i=j}^{K}\left(\frac{N_{\rm t}N_{\rm r}}{4P\sigma^2}{\bigtriangleup x_i^2}+1\right)}\right.+
\nonumber\\&
\left.\frac{\mathbb{B}\left(\sum_{i=1}^{K}\left(\frac{N_{\rm t}N_{\rm r}}{3P\sigma^2}{\bigtriangleup x_i^2}+1\right),c-K+1\right)}{4\prod \limits_{j=2}^{K}\sum_{i=j}^{K}\left(\frac{N_{\rm t}N_{\rm r}}{3P\sigma^2}{\bigtriangleup x_i^2}+1\right)}\right)+\frac{1}{2^\eta}P\left(M_{\rm C}<K\right).    
\end{align}
\end{figure*}
\setcounter{equation}{45}
\end{proposition}
\begin{proof}
    See Appendix \ref{C}.
\end{proof}
\setcounter{equation}{44}
Then the expression of APEP is derived as
 \begin{align}\label{APEP}
 P_{\rm APEP}=&\frac{1}{\eta 2^\eta}\sum_{\bm{\bar{x}}_{\rm C}}\sum_{\bm{\hat{x}}_C}P(\bm{\bar{x}}_{\rm C} \rightarrow \bm{\hat{x}}_C)e(\bm{\bar{x}}_{\rm C}, \bm{\hat{x}}_C),
 \end{align}
where $P(\bm{\bar{x}}_{\rm C} \rightarrow \bm{\hat{x}}_C)$ is given by proposition \ref{pairwise proposition} and $e(\bm{\bar{x}}_{\rm C}, \bm{\hat{x}}_C)$ denotes the number of error bits between $\bm{\bar{x}}_{\rm C}$ and $\bm{\hat{x}}_C$.

{As shown in Eq. (\ref{pairwise}), compared to GBM \cite{gao2019spatial}, the change in APEP originates from the damage caused by additional sensing beams to the effective communication paths. In fact, GBM is a special case of BPM-ISAC when $W=0$, i.e.,
\begin{align}
P(M_{\rm C}=c)=
\left\{\begin{array}{l} 1, c=P
	\\ 0,  c=0,...,P-1.
\end{array}\right.
\end{align}
It is noted that the reduction in effective communication paths lowers the lower bound of APEP, denoted as $P_{\rm APEP}^*$. As the noise approaches zero, APEP of BPM-ISAC approaches its lower bound, given by
\begin{align}\label{lower_bound}
P_{\rm APEP}^*=&\frac{1}{\eta 2^{2\eta}}P\left(M_{\rm C}<K\right)\sum_{\bm{\bar{x}}_{\rm C}}\sum_{\bm{\hat{x}}_C}e(\bm{\bar{x}}_{\rm C}, \bm{\hat{x}}_C).
\end{align}
Therefore, the number of sensing beams determines the boundary of the theoretical BER performance.}

\subsection{CRB Analysis}
To further illustrate the sensing performance of the proposed scheme, the CRB {\cite{van2002optimum}} of DoA estimation is derived. 
Employing the Swerling-II model
\cite{skolnik1980introduction}, the reflection coefficient $\beta_i$ is assumed to be constant during each scanning.
According to Eq. (\ref{sensing baseband}), for $L$ sample times per scanning, the baseband signal $\bm{Y}_{\rm B}=[\bm{y}_{\rm B}^1,...,\bm{y}_{\rm B}^{L}]$ can be derived as
\begin{align}
    \bm{Y}_{\rm B}=\bm{T}_{\rm B}^{\rm H}\Xi\bm{T}_{\rm B}\bm{P}_{\rm S}\bar{\bm{X}}_{\rm S}+\bm{N}_{\rm B},
\end{align}
where $\bar{\bm{X}}_{\rm S}=[\bar{\bm{x}}_{\rm S}^{1},...,\bar{\bm{x}}_{\rm S}^{L}]$.
$\bar{\bm{x}}_{\rm S}^{l}$ and $\bm{y}_{\rm B}^{l}$ are the transmitted sensing signal and baseband received echo signal of the $l\mbox{-}$th sample. $\bm{Y}_{\rm B}$ obeys complex Gaussian distribution $\mathcal{CN}(\bm{M}_{\rm Y},\bm{R}_{\rm Y})$, where $\bm{M}_{\rm Y}=\bm{T}_{\rm B}^{\rm H}\Xi\bm{T}_{\rm B}\bm{P}_{\rm S}\bar{\bm{X}}_{\rm S}$ and $\bm{R}_{\rm Y}=\bm{R}_{\rm B}$.
For the target located in the direction of $\psi_i$, given the directions of other targets, the CRB of its DoA estimation can be obtained as follows (See \cite{van2002optimum}, Section 8.2.3):
\begin{align}\label{CRB}
&\overline{\rm CRB}\left(\psi_i\right)\nonumber\\
=&\left\{-{\rm Tr}\left(\frac{\partial\bm{R}_{\rm Y}^{-1}}{\partial \psi_i}\frac{\partial\bm{R}_{\rm Y}}{\partial \psi_i}\right)+2\Re\left\{{\rm Tr}\left(\frac{\partial\bm{M}_{\rm Y}^{\rm H}}{\partial\psi_i}\bm{R}_{\rm Y}^{-1}\frac{\partial\bm{M}_{\rm Y}}{\partial\psi_i}\right)\right\}\right\}^{-1} \nonumber\\
=&\left\{\hspace{-0.02in}2\Re\hspace{-0.02in} \left\{\hspace{-0.02in}{\rm Tr}\left(\hspace{-0.02in}\frac{\partial\hspace{-0.02in}\left(\bm{T}_{\rm B}^{\rm H}\Xi\bm{T}_{\rm B}\bm{P}_{\rm S}\bar{\bm{X}}_{\rm S}\right)^{\rm H}}{\partial \psi_i}\bm{R}_{\rm B}^{-1}\frac{\partial\bm{T}_{\rm B}^{\rm H}\Xi\bm{T}_{\rm B}\bm{P}_{\rm S}\bar{\bm{X}}_{\rm S}}{\partial \psi_i}\right)\hspace{-0.02in}\right\}\hspace{-0.02in}\right\}\hspace{-0.02in}^{-1}.
\end{align}
Taking the expectation of $\overline{\rm CRB}\left(\psi_i\right)$ with respect to $\bm{\bar{X}}_{\rm S}$ and considering $\bm{R}_{\bm{\bar{x}}_{\rm S}}=\bm{D}$, the final expression is written as
\begin{align}
&{\rm CRB}\left(\psi_i\right)\nonumber\\
=&\frac{1}{2\lvert\beta_i\rvert^2}\left(Tr\left(\bm{P}_{\rm S}^{\rm H}\bm{F}_{\rm S}^{\rm H}\dot{\bm{A}_i^{\rm H}}\bm{F}_{\rm S}\bm{R}_{\rm B}^{-1}\bm{F}_{\rm S}^{\rm H}\dot{\bm{A}_i}\bm{F}_{\rm S}\bm{P}_{\rm S}\bm{D}\right)\right)^{-1},
\end{align}
where $\dot{\bm{A}_i}=\dot{\bm{a}}(\psi_i)\bm{a}^{\rm H}(\psi_i)+\bm{a}(\psi_i)\dot{\bm{a}}^{\rm H}(\psi_i)$.
{Furthermore, denoting matrix $\bm{F}_{\rm S}^{\rm H}\dot{\bm{A}_i^{\rm H}}\bm{F}_{\rm S}\bm{R}_{\rm B}^{-1}\bm{F}_{\rm S}^{\rm H}\dot{\bm{A}_i}\bm{F}_{\rm S}$ as $\bm{M}$, the CRB can be expressed in the following form:
\begin{align}
{\rm CRB}\left(\psi_i\right)=\frac{1}{2\lvert\beta_i\rvert^2\sum_{i=1}^W b_i^2 d_i \bm{M}_{ii}}.
\end{align}
Clearly, $\bm{M}$ is a positive definite matrix, and $M_{ii}\ge 0$. With the fixed activation probability $d_i$ and the analog precoder, a higher transmission power results in a lower CRB. Therefore, sensing beampattern MSE minimization under a given transmission power constraint helps improve the performance of DoA estimation.}
\subsection{Extension to Multiple RF Chains for Sensing}
In the previous modeling, only a single RF chain is dedicatedly spared for sensing.
Indeed, the number of RF chains can be extended to $W_S$, where $1\le W_S\le W$.
In this case, there are $W_S$ out of $W$ beams simultaneously activated, resulting in a total of ${\rm C}_W^{W_S}$ patterns.
The activation probability matrix $\bm{D}$ is no longer a diagonal matrix and is determined by the predefined activation probability of each pattern.
For instance, when $W_S=W$, $\bm{D}$ becomes a matrix filled with 1. 
By substituting the correct matrix $\bm{D}$, the proposed transceiver design can be easily applied to the scenario with multiple RF chains for sensing.
It is worth noting that increasing the number of RF chains for sensing can accelerate scanning speed, improving sensing accuracy, especially in high dynamic scenarios. 
Nevertheless, such an improvement comes at the cost of increased hardware overhead.
Therefore, the selection of the number of sensing RF chains should carefully balance the sensing efficiency and hardware cost.
\vspace{-2mm}
\section{Simulations}
\begin{figure}[t]
    \centering
    \includegraphics[width=8cm]{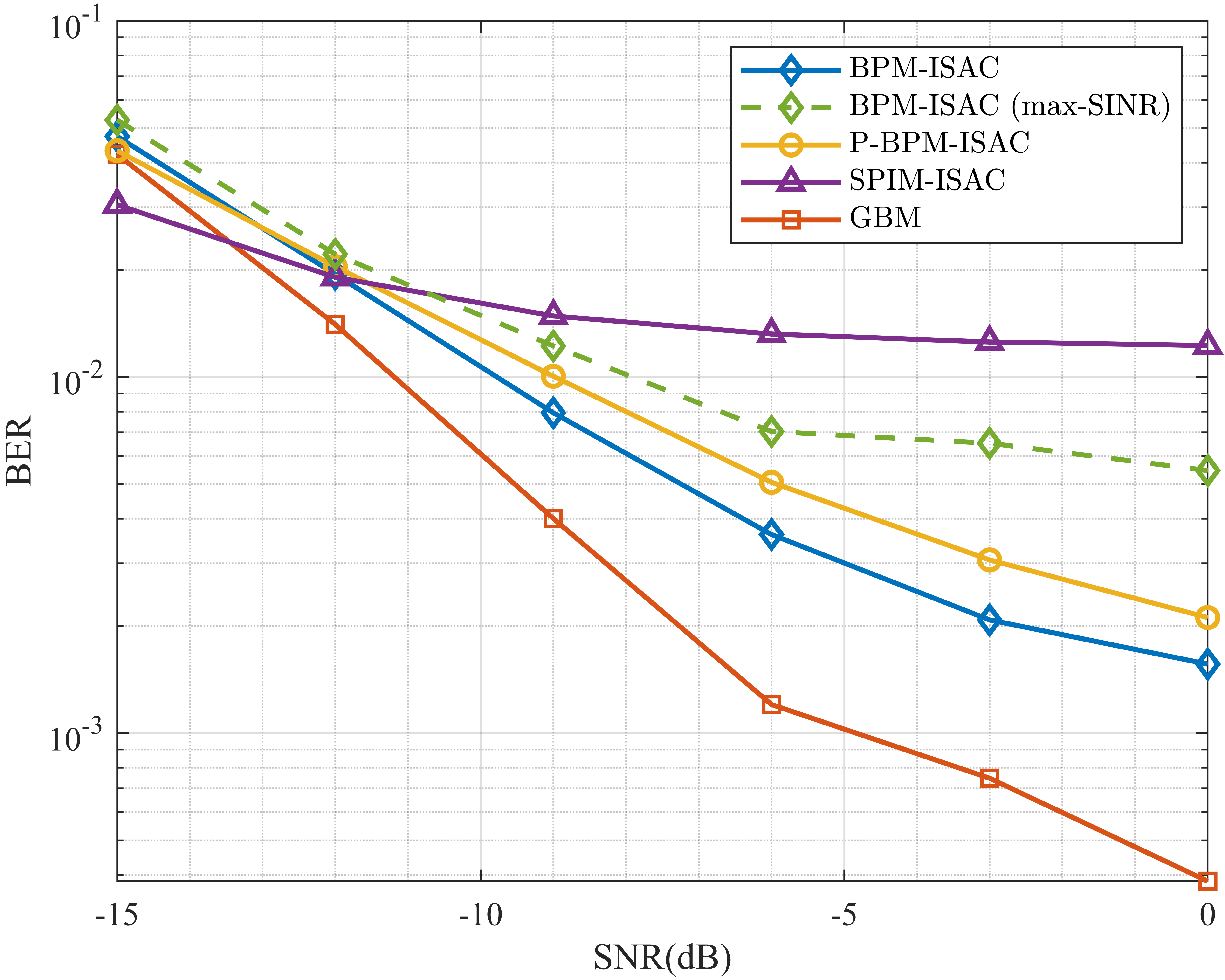}
    \vspace{-3mm}
    \caption{{BER comparison among BPM-ISAC-MBS ($\mu=0.5$), its variants, and other transceiver designs ($\eta=8$ bps/Hz).}}
    \vspace{-3mm}
    \label{S1}
    \end{figure}
\begin{figure}[t]
    \centering
    \includegraphics[width=8cm]{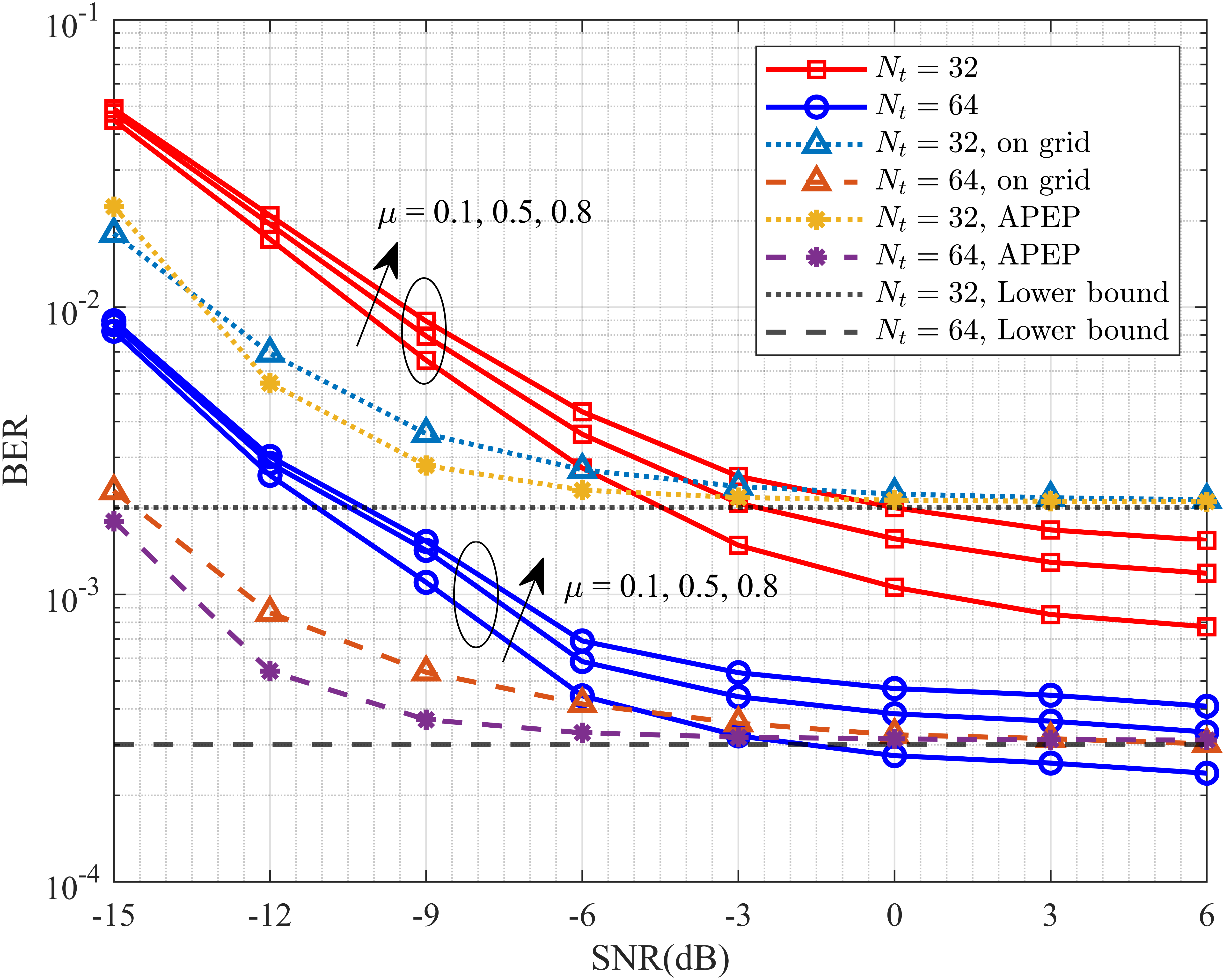}
    \vspace{-3mm}
    \caption{{BER performance of BPM-ISAC-MBS ($N_{\rm r}=32$).}}
    \label{S2}
    \end{figure}
\begin{figure}[t]
    \centering
    \includegraphics[width=8cm]{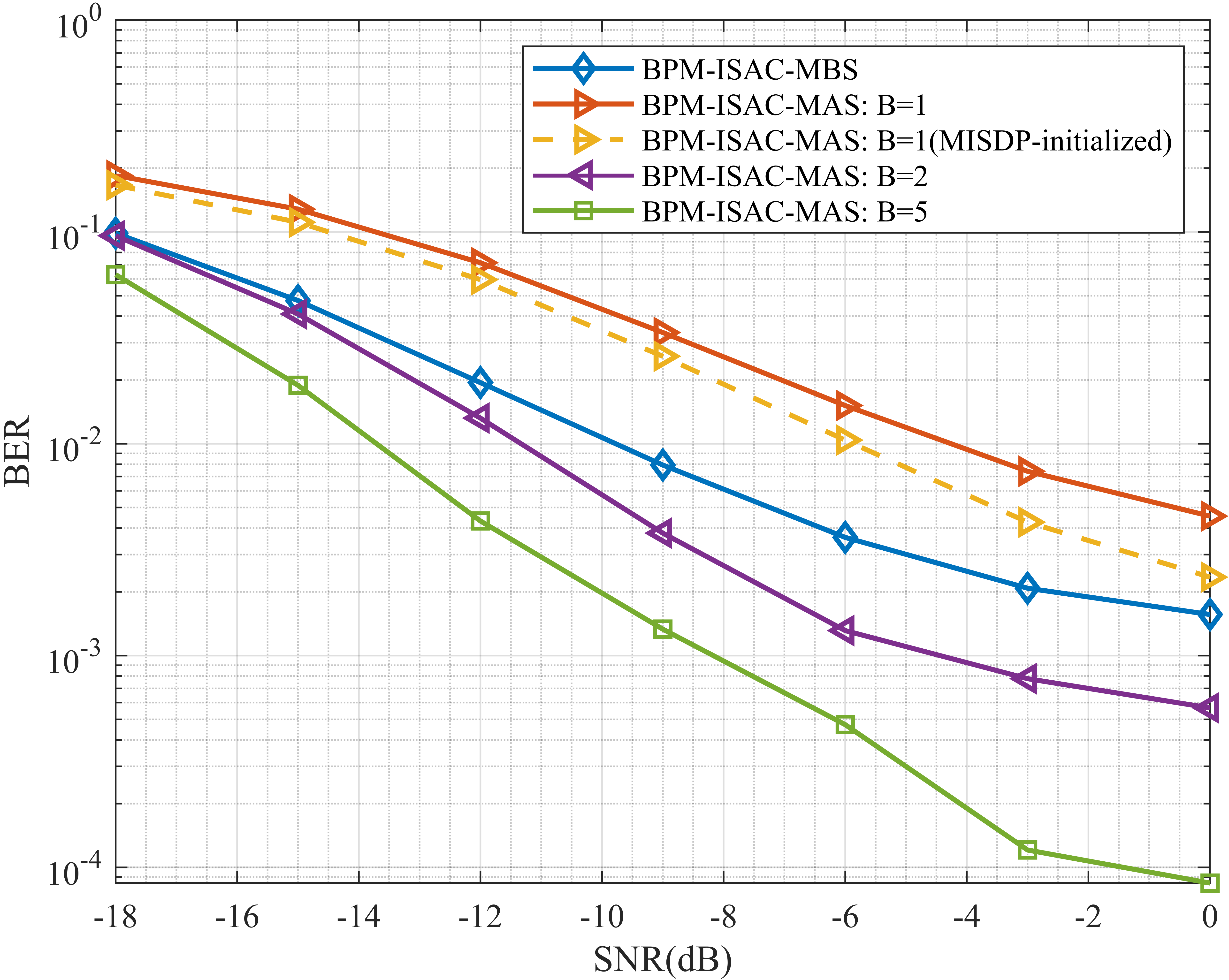}
    \vspace{-3mm}
    \caption{{BER performance of BPM-ISAC-MBS and BPM-ISAC-MAS ($\mu$=0.5).}}
    \vspace{-3mm}
    \label{S3}
    \end{figure}
\begin{figure}[t]
    \centering
    \includegraphics[width=7.9cm]{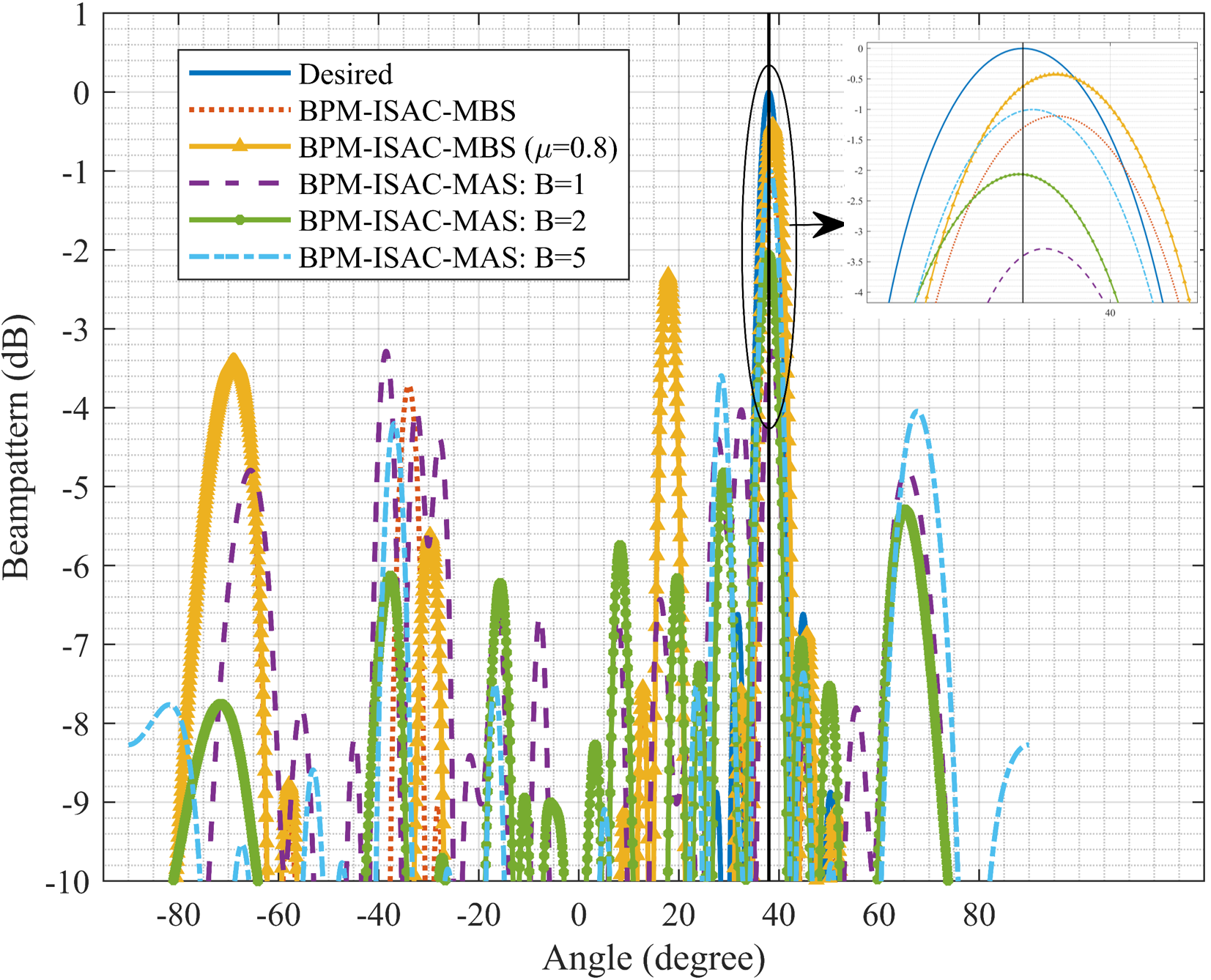}
    \vspace{-3mm}
    \caption{{Instantaneous normalized beampattern with BPM-ISAC-MBS and BPM-ISAC-MAS ($\mu$=0.5 or 0.8, sensing beam points at $38^{\circ}$).}}
    \label{S4}
    \vspace{-3mm}
    \end{figure}
In this section, we evaluate the communication and sensing performance of the proposed BPM-ISAC method through numerical simulation. 
We consider a hybrid mmWave ISAC system, where $N_{\rm t}=N_{\rm r}=N_{\rm e}=32$ {unless otherwise specified}.
Suppose there are $P=8$ non-line-of-sight (NLoS) paths with $\alpha_i \sim \mathcal{CN}(0,1)$, and $\theta_i$ and $\phi_i$ are uniformly distributed in $[-\pi/2,\pi/2)$. 
For communication, we adopt 4-QAM modulation and set $K=4$, $N_{\rm C}=3$, and $L=20$. 
For sensing, we set $W=3$ and  $T_{\rm R}=5$.
Without loss of generality, we assume two targets are located at $\psi_1=39^{\circ}$ and $\psi_2=43^{\circ}$ with reflection coefficients of $\lvert\beta_1\rvert=\lvert\beta_2\rvert=1$.
The scanning directions of interest is set as $[38^{\circ},44^{\circ},50^{\circ}]$.
The ideal beampattern is $\bm{t}=\sqrt{T_{\rm R}}\bm{1}_W$ and the activation probability matrix is $\bm{D}=\frac{1}{W}\bm{I}_{W}$.
For algorithms 2 and 3, convergence tolerance is set as $0.001$, and the maximum number of iterations is set as 50. 
The signal-to-noise ratio (SNR) is defined as $\frac{E_b}{N_0}=\frac{N_{\rm C}}{\eta\sigma^2}$.

To simplify the representation, `BPM-ISAC-MBS' and `BPM-ISAC-MAS' denote our proposed method with MBS and MAS, respectively. 
For comparison, some relevant methods and variants are introduced.
`SPIM-ISAC' refers to \cite{elbir2023millimeter} which utilizes $K$ strongest spatial paths for communication and 
`GBM' refers to \cite{gao2019spatial}.
`P-BPM-ISAC' denotes the plain version of `BPM-ISAC-MBS', which utilizes $K$ beams simultaneously without index modulation.
`BPM-ISAC-MBS' with maximum SINR-based beam selection criterion is also presented, i.e., the beam pairs with the largest signal-to-interference-plus-noise-ratio (SINR) are selected, where the SINR of beam pairs $(i,j)$ is defined as 
\begin{align}
\textrm{SINR}[i,j]=\frac{\lvert\bm{\bar{H}}[i,j]\rvert^2}{\sum_{k\in \Omega}\lvert\bm{\bar{H}}[i,k]\rvert^2+\sigma^2}.
\end{align} 
For `EDC-ISAC', fully digital architecture is adopted and eigenvectors corresponding to $K$ largest eigenvalue of the spatial channel are utilized to construct EDC.
\vspace{-3mm}
\subsection{Communication Performance}
In Fig. \ref{S1}, we compare the BER performance of BPM-ISAC-MBS with $\mu=0.5$, its variants, and other schemes.
For a fair comparison, all schemes adopt the 4-QAM modulation to keep the same SE as 8 bps/Hz.
SPIM-ISAC \cite{elbir2023millimeter} exhibits high BER at high SNR due to severe sensing interference.
BPM-ISAC-MBS with max-SINR beam selection criterion performs worse, indicating the advantage of the proposed min-MSE criterion. 
In high SNR regions, BPM-ISAC demonstrates lower BER than P-BPM-ISAC, highlighting the superiority of beam pattern modulation. 
In addition, the performance of GBM is provided as a reference, which is the special case of BPM-ISAC-MBS without sensing interference.
    
In Fig. \ref{S2}, the BER performance of BPM-ISAC-MBS with different $N_{\rm t}$ and $\mu$ is presented.
As $\mu$ increases, strengthening the communication constraint, the BER performance gradually decreases.
The BER performance of $N_{\rm t}=64$ is better than the case of $N_{\rm t}=32$ due to the array gain.  
The BER performance with the on-grid beamspace channel and unoptimized digital precoder is presented, which is consistent with APEP analysis at high SNR regions.
{The APEP lower bound is also presented according to Eq. (\ref{lower_bound}).}
It can be observed that the on-grid case has better BER performance than the normal case at the low SNR region.
This is due to that Gaussian noise is the main interference factor at low SNR and the communication beams of the on-grid case have more concentrated energy without beam leakage.
At high SNR, sensing interference becomes the main interference and these two cases perform similarly.

In Fig. \ref{S3}, the BER performance of BPM-ISAC-MBS and BPM-ISAC-MAS with different bit resolutions are illustrated.
For BPM-ISAC-MAS, the BER decreases with the increase in bit number due to the higher freedom degree of the optimized beam pattern.
In addition, the BER performance of MISDP-initialized BPM-ISAC-MAS with 1-bit PSs is presented.
To reduce computation time, $\bm{F}_{\rm C}$ and $\bm{W}_{\rm RF}$ have been optimized only once alternatively.
It is observed that with proper initialization, the BER of 1-bit BPM-ISAC-MAS approaches BPM-ISAC-MBS at high SNR.
\vspace{-3mm}
\subsection{Sensing Performance}
\begin{figure}[t]
    \centering
    \includegraphics[width=8cm]{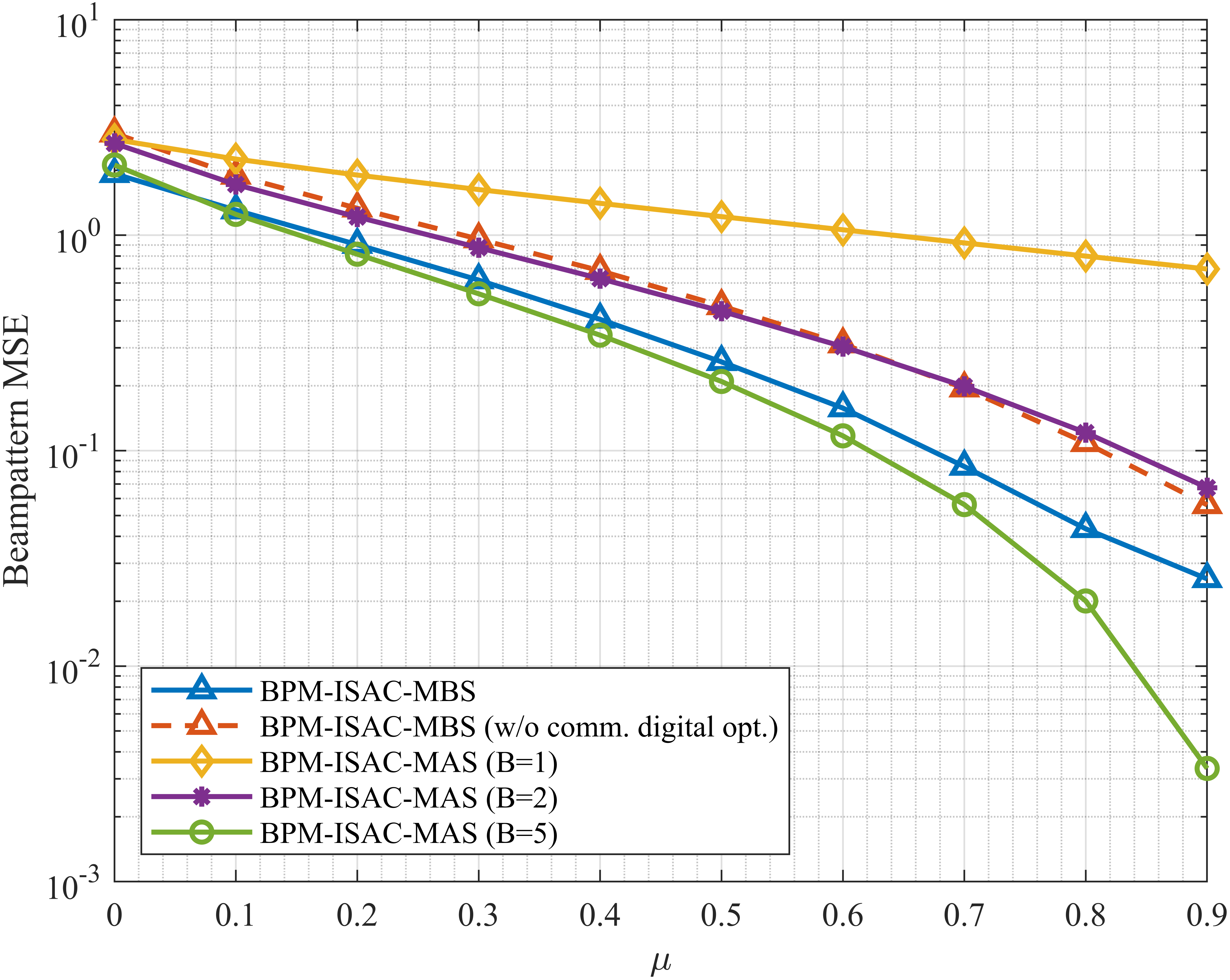}
    \vspace{-3mm}
    \caption{{Sening beampattern MSE of BPM-ISAC-MBS and BPM-ISAC-MAS.}}
    \label{S5}
    \end{figure}
\begin{figure}[ht]
    \centering
    \includegraphics[width=8cm]{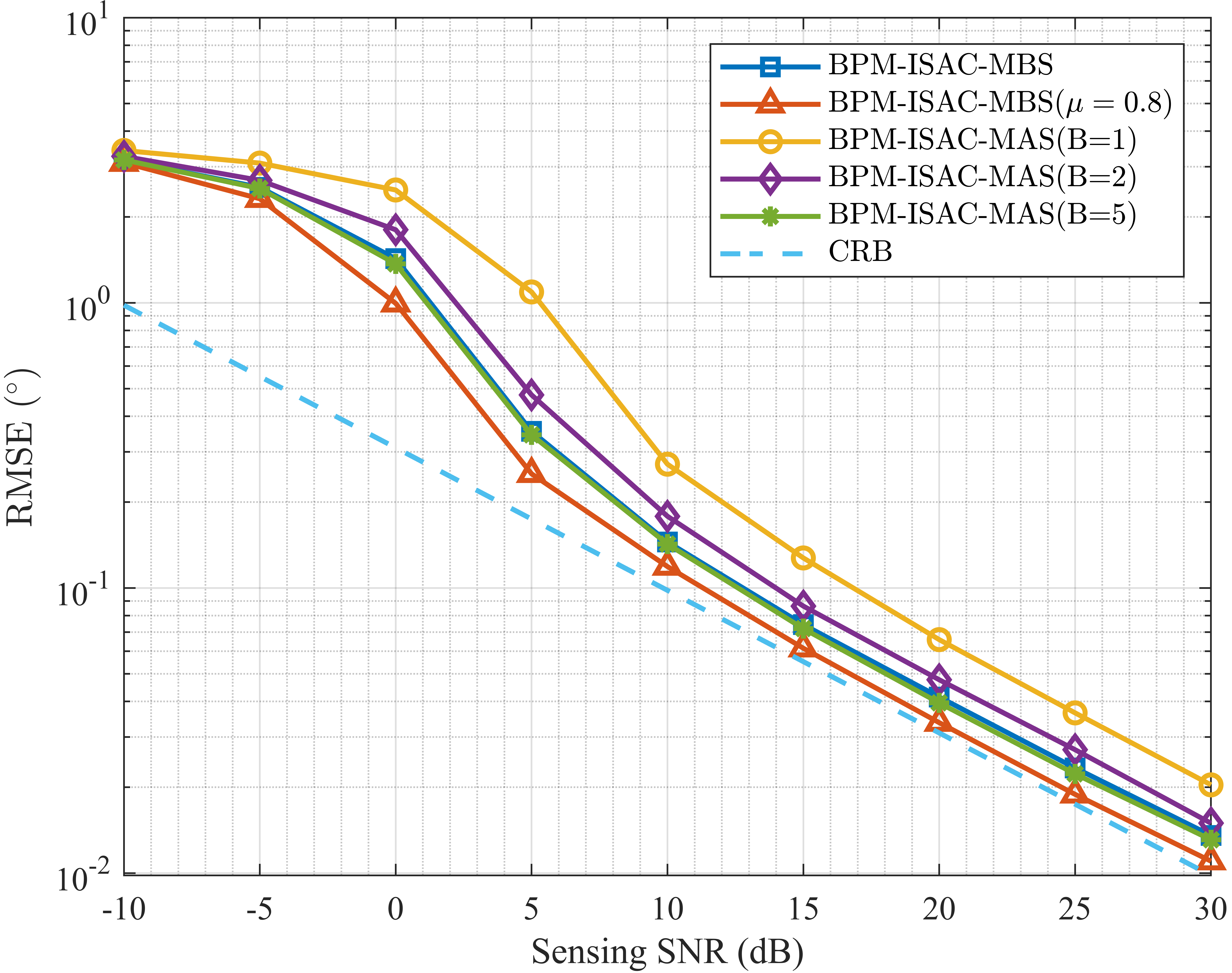}
    \vspace{-3mm}
    \caption{The RMSE performance of BMUSIC-based DoA estimation algorithm versus sensing SNR ($\mu=0.5$ or 0.8).}
    \label{S6}
    \vspace{-4mm}
\end{figure}
\begin{figure}[ht]
    \centering
    \includegraphics[width=8cm]{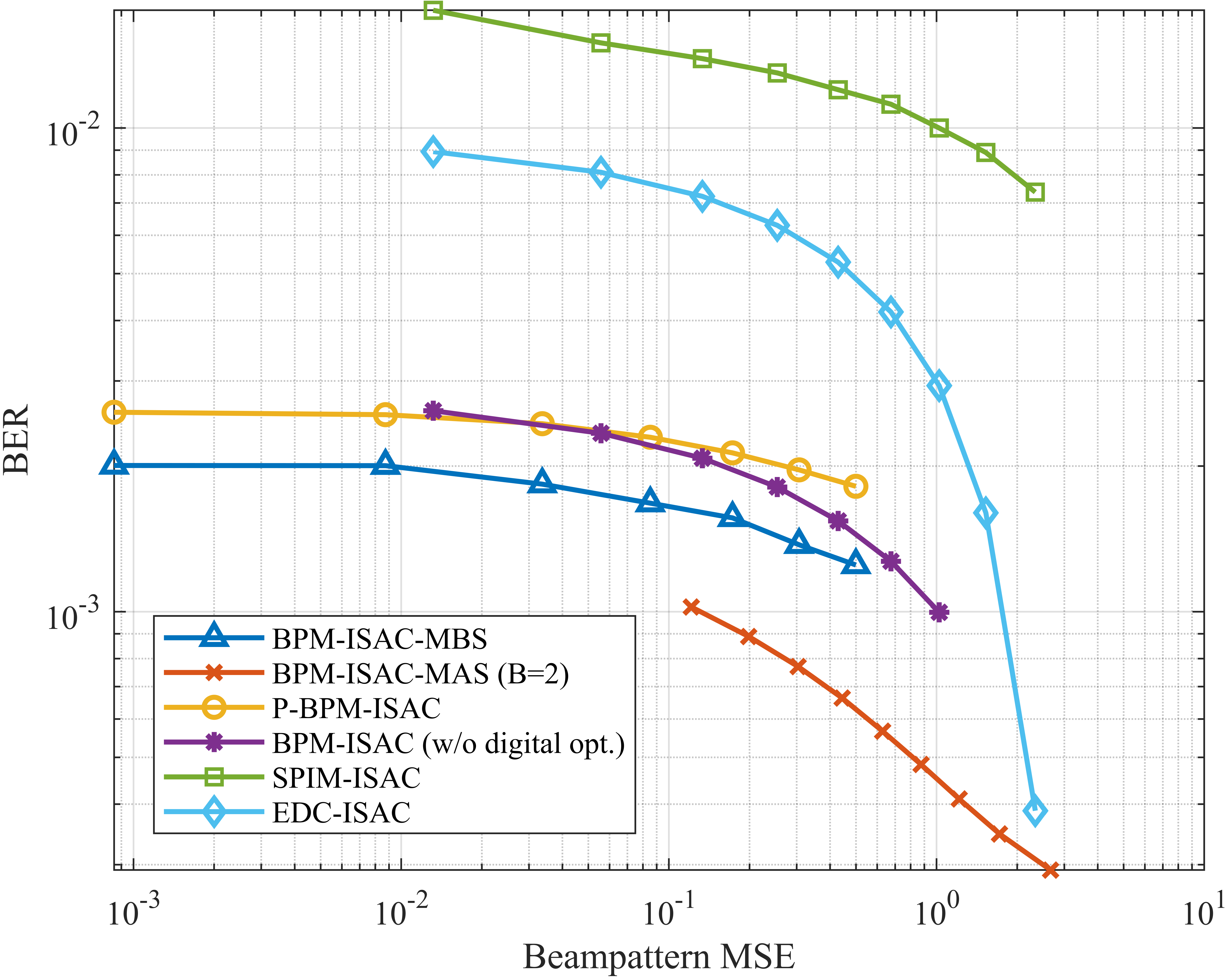}
    \vspace{-3mm}
    \caption{{BER and beampattern MSE performance trade-off among different ISAC transceiver designs (SNR = 0 dB).}}
    \label{S8}
\end{figure}
In Fig. \ref{S4}, we present the normalized beampattern of the proposed method at a certain moment when the sensing beam pointing at $38^{\circ}$.
It can be observed that the strongest beam points in the direction of interest, while multiple other beams are activated for communication.
Due to the discrete codewords, there exists a certain deviation from the desired direction for BPM-ISAC-MBS, which can be neglected for massive antennas.
Compared with BPM-ISAC-MBS based on beamspace, BPM-ISAC-MAS offers a more flexible beam pattern, enhancing the equality of the equivalent digital channel.

In Fig. \ref{S5}, the beampattern MSE versus weighting coefficient $\mu$ is presented to illustrate the beampattern performance of the proposed method under different $\mu$ values.
The beampattern MSE decreases with the increase of $\mu$ because the augmented communication constraint compromises the power allocation of sensing beams. 
In addition, the BPM-ISAC-MBS without communication digital precoder optimization has a higher beampattern MSE. 
This is because optimized communication power allocation can improve communication performance and  {implicitly} relax the constraint on sensing power.

To further validate the sensing performance of the proposed method, root mean square error (RMSE) of DoA estimation versus sensing SNR using beamspace MUSIC algorithm \cite{lee1990resolution} is shown in Fig. \ref{S6}. 
The sensing SNR is defined as the ratio between the $T_{\rm R}$ and the noise power of $\bm{\xi}_{\rm R}$.
It can be observed that, at high SNR, there are different gaps between the RMSE of DoA estimation and the ideal CRB defined in Eq. (\ref{CRB}). 
This is due to the varying degrees of suppression of sensing power under different constraints.
The performance of DoA estimation is generally consistent with the beampattern performance, indicating the effectiveness of choosing the beampattern as the sensing performance metric. 
\vspace{-3mm}
\subsection{Communication and Sensing Trade-off}
In Fig. \ref{S8}, the communication and sensing trade-off curves between BER and beampattern MSE among different schemes are presented for fair comparison.
Within the testing scope, BPM-ISAC consistently outperforms other alternatives.
It is notable that for large beampattern MSE, i.e., the sensing power is limited, the EDC-ISAC scheme achieves similar BER performance as BPM-ISAC-MBS with 2-bit PSs.
However, as the sensing power increases, the BER performance of EDC-ISAC and SPIM-ISAC sharply deteriorates, whereas the proposed scheme demonstrates significant advantages thanks to effective optimization.
BPM-ISAC-MBS with 2-bit PSs demonstrates an advantage over BPM-ISAC-MAS due to the higher degree of freedom of analog precoders.
In addition, the performance of BPM-ISAC without digital-part optimization and P-BPM-ISAC are provided to demonstrate the effectiveness of power allocation and beam pattern modulation, respectively.

\section{Conclusions}
In this paper, we have proposed a novel beam pattern modulation embedded mmWave ISAC hybrid transceiver design, termed BPM-ISAC. 
BPM-ISAC aims to retain the SE benefits of primitive beamspace modulation schemes while addressing performance bottlenecks in their extension to  ISAC functionalities. 
To ensure near-optimal performance for BPM-ISAC, we formulated an optimization problem to minimize the sensing beampattern MSE under the symbol MSE constraint and solved it by optimizing analog and digital parts sequentially. 
Both the MBS and MAS hybrid structures are considered for analog configurations.
Theoretical analysis and simulation results have verified that the proposed BPM-ISAC offers an overall improved trade-off in sensing and communication performance.

\begin{appendices}
\section{Proof of Proposition 1}\label{A}
At high SNR, $\bm{W}_{\rm BB,0}\simeq\bm{H}_{\rm C}^\dagger=(\bm{H}_{\rm C}^{\rm H}\bm{H}_{\rm C})^{-1}\bm{H}_{\rm C}^{\rm H}$. Thus 
$\frac{N_{\rm C}}{K}\Vert\bm{W}_{\rm BB,0}\bm{H}_{\rm C}-\bm{I}_K\Vert_F^2\simeq 0$ and the objective function is simplified as
\begin{align}\label{simp Gamma}
    \overline{\rm MSE}_{\rm C}\simeq &{\rm Tr}\bigg(\left(\bm{W}_{\rm RF}^{\rm H}\bm{H}\bm{F}_{\rm C}{\rm{diag}}\left(\bm{t}\right)^2\bm{D}\bm{F}_{\rm S}\bm{H}^{\rm H}\bm{W}_{\rm RF}+\sigma^2\bm{I}_K\right)\nonumber\\
&\left.\left(\bm{W}_{\rm RF}^{\rm H}\bm{H}\bm{F}_{\rm C}\bm{F}_{\rm C}^{\rm H}\bm{H}^{\rm H}\bm{W}_{\rm RF}\right)^{-1}\right).
\end{align}

Let $\overline{\rm MSE}_{\rm C}={\rm Tr}\left(\frac{w}{K+T}\bm{I}_{K+T}\right)$.
With the Schur complement \cite{chen2020new}, it can be proved \cite{di2020hybrid} that minimizing $\overline{\rm MSE}_{\rm C}$ is equivalent to minimizing $w$ and $\mathcal{P}.1\mbox{-}2$ can be reformulated as shown in proposition 1.
\vspace{-3mm}
\section{Proof of Proposition 2}
\label{LMI}
Let $\bar{\bm{F}}=\bm{F}_{\rm C}\bm{F}_{\rm C}^{\rm H}$. For each element, we have
\begin{align}
    \bar{\bm{F}}[i,j]=&\sum_{k=1}^K\bm{F}_{\rm C}[i,k]\overline{\bm{F}_{\rm C}[j,k]}\nonumber\\
=&\sum_{k=1}^Ke^{j(\angle\bm{F}_{\rm C}[i,k]-\angle\bm{F}_{\rm C}[j,k])}\nonumber\\
=&\sum_{k=1}^K \cos \bigtriangleup \theta_{i,k,j,k}+j\sin\bigtriangleup \theta_{i,k,j,k}\nonumber\\
=&\sum_{k=1}^K \bm{c}^{\rm T}\bm{y}^{i,k,j,k}+j\bm{s}^{\rm T}\bm{y}^{i,k,j,k},
\end{align}

\noindent where $\bigtriangleup \theta_{i,k,j,k}=\angle\bm{F}_{\rm C}[i,k]-\angle\bm{F}_{\rm C}[j,k]$. Thus $\bm{F}_{\rm C}\bm{F}_{\rm C}^{\rm H}$ is transformed into the linear function of $\bm{y}^{i,j,i^{\prime},j^{\prime}}$.

\section{Proof of Proposition 3}\label{KKT_proof} 
{
The lagrange function of $\mathcal{P}.2$ is derived as
\begin{align}
L(\bm{p},\bm{b},\lambda_1,\lambda_2,\lambda_3)&=\sum_{i=1}^W d_i(\lvert\bm{a}_{N_{\rm t}}^{\rm H}(\theta_i)\bm{F}_{\rm S}[:,i]\rvert b_i-t_i)^2+\nonumber\\
\lambda_1&(\Vert{\rm{diag}}(\bm{p})\Vert_F^2-K)+\lambda_2(\Vert\bm{D}^{\frac{1}{2}}\bm{b}\Vert_F^2 - T_{\rm R})\nonumber\\
&+\lambda_3({\rm MSE_C}(\bm{p},\bm{b})-\Gamma)
\end{align}
Denote the convergence solution of $\mathcal{P}.2$ with alternating optimization algorithm as $\bm{P}_{\rm C}=\bm{p}^*$ and $\bm{P}_{\rm S}=\bm{p}^*$.
Denote $\bm{W}_{\rm BB}(\bm{p}^*,\bm{b}^*)$ as $\bm{W}_{\rm BB}^*$.
The KKT conditions of $\mathcal{P}.2$ are given by
\begin{subequations}
\begin{align}
&\frac{\partial L}{\partial \bm{p}}\bigg|_{\bm{p}=\bm{p}^*}=\bm{0},  \frac{\partial L}{\partial \bm{b}}\bigg|_{\bm{b}=\bm{b}^*}=\bm{0}\label{KKT_1},\\
&\Vert{\rm{diag}}(\bm{p}^*)\Vert_F^2\le K,  \Vert\bm{D}^{\frac{1}{2}}\bm{b}^*\Vert_F^2 \le T_{\rm R},  {\rm MSE_C}(\bm{p}^*,\bm{b}^*)\le\Gamma\label{KKT_2},\\
& \lambda_1\ge 0, \lambda_2\ge 0, \lambda_3\ge 0\label{KKT_3},\\
& \lambda_1(\Vert{\rm{diag}}(\bm{p}^*)\Vert_F^2- K)=0,  \lambda_2(\Vert\bm{D}^{\frac{1}{2}}\bm{b}^*\Vert_F^2 - T_{\rm R})=0, \nonumber\\
&\lambda_3({\rm MSE_C}(\bm{p}^*,\bm{b}^*)-\Gamma)=0\label{KKT_4}
\end{align}
\end{subequations}
For the convex QCQP problem (\ref{QCQP}) of the step 1), ($\bm{p}^*,\bm{b}^*$) is its optimal solution and naturally satisfies KKT conditions. 
Different from $\mathcal{P}.2$, $\bm{W}_{\rm BB}$ in Problem (\ref{QCQP}) is fixed as constant matrix $\bm{W}_{\rm BB}^*$.
It is clear that in this case, constraint (\ref{KKT_2}), (\ref{KKT_3}), and (\ref{KKT_4}) are satisfied. 
In addition, the stationarity condition satisfies
\begin{equation}
\frac{\partial L(\bm{W}_{\rm BB}=\bm{W}_{\rm BB}^*)}{\partial \bm{p}}\bigg|_{\bm{p}=\bm{p}^*}=\bm{0}, \frac{\partial L(\bm{W}_{\rm BB}=\bm{W}_{\rm BB}^*)}{\partial \bm{b}}\bigg|_{\bm{b}=\bm{b}^*}=\bm{0}
\end{equation}
Note that condition (\ref{KKT_1}) holds only when the following condition is satisfied.
\begin{subequations}
\begin{align}
&\frac{\partial {\rm MSE_C}(\bm{p})}{\partial \bm{p}}\bigg|_{\bm{p}=\bm{p}^*}=\frac{\partial {\rm MSE_C}(\bm{p},\bm{W}_{\rm BB}=\bm{W}_{\rm BB}^*)}{\partial \bm{p}}\bigg|_{\bm{p}=\bm{p}^*},\label{KKT_5}\\
&\frac{\partial {\rm MSE_C}(\bm{b})}{\partial \bm{b}}\bigg|_{\bm{b}=\bm{b}^*}=\frac{\partial {\rm MSE_C}(\bm{b},\bm{W}_{\rm BB}=\bm{W}_{\rm BB}^*)}{\partial \bm{b}}\bigg|_{\bm{b}=\bm{b}^*}.\label{KKT_6}
\end{align}
\end{subequations}
Taking (\ref{KKT_5}) as an example. 
According to the chain rule of differentiation, we have
\begin{align}
&\frac{\partial {\rm MSE_C}(\bm{p})}{\partial \bm{p}}\bigg|_{\bm{p}=\bm{p}^*}=\frac{\partial {\rm MSE_C}(\bm{p},\bm{W}_{\rm BB}=\bm{W}_{\rm BB}^*)}{\partial \bm{p}}\bigg|_{\bm{p}=\bm{p}^*}+\nonumber\\
& \frac{\partial {\rm MSE_C}(\bm{W}_{\rm BB})}{\partial \bm{W}_{\rm BB}}\bigg|_{(\bm{p},\bm{W}_{\rm BB})=(\bm{p}^*,\bm{W}_{\rm BB}(\bm{p}^*))}\frac{\partial {\rm MSE_C}(\bm{W}_{\rm BB})}{\partial \bm{p}}\bigg|_{\bm{p}=\bm{p}^*}.\label{chain}
\end{align}
Since Eq. (\ref{m}) is satisfied for convergence solution, $\bm{W}_{\rm BB}$ is the LMMSE equalizer to minimize ${\rm MSE_C}$ and satisfies
\begin{align}
\frac{\partial {\rm MSE_C}(\bm{W}_{\rm BB})}{\bm{W}_{\rm BB}}=0.
\end{align}
Thus the second term on the right side of Eq. (\ref{chain}) is 0 and condition (\ref{KKT_5}) is satisfied.
Similarly, condition (\ref{KKT_6}) is also satisfied.
Then condition (\ref{KKT_1}) is satisfied.
Therefore, the convergence point satisfies the KKT conditions of $\mathcal{P}.2$.}

\section{Proof of Proposition 4}
\label{B}Let $P_{M_{\rm R}}(r)$ and  $P_{M_{\rm B}}(r,b)$  represent the probability that sensing beams cover $M_{\rm R}=r$ paths and these paths cover $M_{\rm B}=b$ received beams, respectively. 
Let $P_{M_{\rm C}}(c,r,b)$ represents the probability that $M_{\rm C}=c$ paths are available for communication when $M_{\rm R}=r$ and $M_{\rm B}=b$. 
Then the probability distribution of $M_{\rm C}$ can be easily obtained as Eq. (\ref{P com}).
Both $P_{M_{\rm R}}(r)$ and $P_{M_{\rm C}}(c,r,b)$ belong to the classical probability model and can be derived as Eq. (\ref{P_MR}) and Eq. (\ref{P_MC}) using the combination number formula.
For $P_{M_{\rm B}}(r,b)$, we can obtain it through a recursive process as Eq. (\ref{P_MB}).
\section{Proof of Proposition 5}
\label{C}
Supposing $M_{\rm C}<K$, effective communication cannot be achieved and BER is set to 0.5.
When $M_{\rm C}\geq K$, denote $\gamma_i=\frac{P}{N_{\rm t}N_{\rm r}}\bm{H}_{\rm C}^2[i,i]$ and $\bigtriangleup x_i=\bm{\bar{x}}_{\rm C}[i]-\bm{\hat{x}}_C[i]$, and then the pairwise error probability is given as 
\begin{align}\label{Q_f}
\hspace{-0.2in}&P(\bm{\bar{x}}_{\rm C} \rightarrow \bm{\hat{x}}_C)\nonumber\\
=&
\sum\limits_{c=K}^P
\mathbb{E}_{M_{\rm C}=c}\hspace{-0.02in}\left\{Q(\sqrt{\frac{\Vert\bm{H}_{\rm C}(\bm{\bar{x}}_{\rm C}-\bm{\hat{x}}_C)\Vert_2^2}{2\sigma^2}})\right\}\hspace{-0.04in}+\hspace{-0.04in}\frac{1}{2^\eta}P(M_{\rm C}<K)\nonumber\\
\overset{(b)}{\simeq} &\sum_{c=K}^P \mathbb{E}_{M_{\rm C}=c}\left\{\frac{1}{12}{\rm{exp}}(-\frac{N_{\rm t}N_{\rm r}}{4P\sigma^2}\sum_{i=1}^K\gamma_i^2\bigtriangleup x_i^2)\right.\nonumber\\
&+\frac{1}{4}\left.{\rm{exp}}(-\frac{N_{\rm t}N_{\rm r}}{3P\sigma^2}\sum_{i=1}^K\gamma_i^2\bigtriangleup x_i^2)
\right\}+\frac{1}{2^\eta}P(M_{\rm C}<K),
\end{align}
where $(b)$ is for that $Q(x)\simeq\frac{1}{12}e^{-\frac{x^2}{2}}+\frac{1}{4}e^{-\frac{2x^2}{3}}$. 
According to Eq. (\ref{channel}), $\gamma_i$ follows a unit exponential distribution. Assume that $K$ out of $c$ largest paths are selected, satisfying $\gamma_1<\gamma_2\cdots<\gamma_K$. Thus the probability distribution of $\bm{\gamma}=[\gamma_1,\cdots,\gamma_K]^{\rm T}$ is given by
\begin{align}
f(\bm{\gamma})=\frac{c!}{(c-K)!}(1-e^{-\gamma_1})^{c-K}
\prod\limits_{i=1}^Ke^{-\gamma_i}.
\end{align}
Then the first item of Eq. (\ref{Q_f}) is derived as
\begin{align}
&\int_{0}^{+\infty}\int_{\gamma_1}^{+\infty}\cdots\int_{\gamma_{K-1}}^{+\infty}f(\gamma)e^{-\frac{N_{\rm t}N_{\rm r}}{4P\sigma^2}\sum\limits_{i=1}^K\gamma_i^2\bigtriangleup x_i^2}d\gamma_1\cdots\gamma_K\nonumber\\
=&\frac{c!}{(c-K)!\prod \limits_{j=2}^{K}n_j}\int_{0}^{+\infty}e^{-\gamma_1n_1}(1-e^{-\gamma_1})^{c-K}d\gamma_1\nonumber\\
=&\frac{c!}{(c-K)!\prod \limits_{j=2}^{K}n_j}\mathbb{B}(n_1,c-K+1),
\end{align}
where $n_j=\sum_{i=j}^{K}(\frac{N_{\rm t}N_{\rm r}}{4P\sigma^2}{\bigtriangleup x_i^2}+1)$ and $\mathbb{B}(p,q)=\int_{0}^1x^{p-1}(1-x)^{q-1}dx$ is the Beta function \cite{gradshteyn2014table}. 
Similarly, we can obtain the second item of (\ref{Q_f}). Then, the pairwise error probability arrives at Eq. (\ref{pairwise}).
\end{appendices}
\bibliographystyle{IEEEtran}
\small
\bibliography{IEEEabrv, ref}

\end{document}